%% file: BCGauss-CN.tex
\documentclass[10pt]{article}

\input{Chead.tex}

\input{defns.tex}

\usepackage{wasysym}

\usepackage{hyperref}

\def\caep{ \mathcal{T}^{(n)}}

\def\psd{\succeq} 
\def\nsd{\preceq} 
\def\cv{\mathbf{c}}
\def\lav{{\vec{\lambda}}}
\def\wc{\stackrel{w}{\Rightarrow}}

\title{The capacity region of the two-receiver vector Gaussian broadcast channel with private and common messages}
\author{Yanlin Geng, Chandra Nair}

\begin{document}

\maketitle

\begin{abstract}
We develop a new method for showing the optimality of the Gaussian distribution in multiterminal information theory problems. As an application of this method we show that Marton's inner bound achieves the capacity of the vector Gaussian broadcast channels with common message.
\end{abstract}

\section{Introduction}

Channels with additive Gaussian noise are a commonly used model for wireless communications. Hence computing the capacity regions or bounds on the capacity regions for these classes of channels are of wide interest. Usually these bounds or capacity regions are represented using auxiliary random variables and distributions on these auxiliary random variables. Evaluations of these bounds then becomes an optimization problem of computing the extremal auxiliary random variables. In several instances involving Gaussian noise channels, it turns out that the optimal auxiliaries and the inputs are Gaussian. However proving the optimality of Gaussian distributions are usually very cumbersome and involve certain non-trivial applications of the entropy-power-inequality(EPI), and the perturbation ideas behind its proof.

For the two-receiver vector Gaussian broadcast channel with private messages, the capacity region was established\cite{wss06} by showing that certain inner and outer bounds match. This argument was indirect, and hence the approach has been hard to generalize to other situations. In the following sections we develop a novel way of proving the optimality of Gaussian input distribution for additive Gaussian noise channels. There are many potential straightforward applications of this new approach which will yield new results as well as recover the earlier results in a simple manner. For the purpose of this article, we will restrict ourselves to two-receiver vector Gaussian channels. We will recover the known results for the private messages case and obtain the capacity region in the presence of a common message as well.

\subsection{Preliminaries}

Broadcast channel\cite{cov72} refers to a communication scenario where a single sender, usually denoted by $X$, wishes to communicate independent messages $(M_0, M_1, M_2)$ to two receivers $Y_1, Y_2$. The goal of the communication scheme is to enable receiver $Y_1$ to recover messages $(M_0, M_1)$  and receiver $Y_2$ to recover messages $(M_0, M_2)$; both events being required to occur with high probability. For introduction to the broadcast channel problem and a summary of known work one may refer to Chapters 5, 8, and 9 in \cite{elk12}.

A broadcast channel is characterized by a probability transition matrix $\qmf(y_1,y_2|x)$. The following broadcast channel is referred to as the vector additive Gaussian broadcast channel
\begin{align*}
\Yv_1 &= G_1 \Xv + \Zv_1 \\
\Yv_2 &= G_2 \Xv + \Zv_2.
\end{align*}
In the above $\Xv \in \Real^t$,  $G_1, G_2$ are $t \times t $ matrices, and $\Zv_1, \Zv_2$ are Gaussian vectors independent of $\Xv$.

\begin{remark}
We assume, w.l.o.g. that $Z_1, Z_2 \sim \Nc(0,I)$. 
\end{remark}

A product broadcast channel is a broadcast channel whose transition probability has the form $\qmf_1(\Yv_{11},\Yv_{21}|\Xv_1) \times \qmf_2(\Yv_{12},\Yv_{22}|\Xv_2)$.  A vector additve Gaussian product  broadcast channel can be represented as
\begin{align*}
\left[ \begin{array}{c} \Yv_{11} \\ \Yv_{12} \end{array} \right] & = \left[ \begin{array}{cc} G_{11} & 0 \\ 0 & G_{12} \end{array} \right]  \left[ \begin{array}{c} \Xv_{1} \\ \Xv_{2} \end{array} \right]  + \left[ \begin{array}{c} \Zv_{11} \\ \Zv_{12} \end{array} \right]  \\
\left[ \begin{array}{c} \Yv_{21} \\ \Yv_{22} \end{array} \right] & = \left[ \begin{array}{cc} G_{21} & 0 \\ 0 & G_{22} \end{array} \right]  \left[ \begin{array}{c} \Xv_{1} \\ \Xv_{2} \end{array} \right]  + \left[ \begin{array}{c} \Zv_{21} \\ \Zv_{22} \end{array} \right].
\end{align*}
In the above  $\Zv_{11}, \Zv_{12}, \Zv_{21}, \Zv_{22}$ are independent Gaussian vectors, also independent of $\Xv_1, \Xv_2$.
\begin{remark}
In this paper we assume that all our channel gain matrices are invertible. Since the set of all matrices are dense (with respect to say, Frobenius norm)  by continuity, our capacity results extend to non-invertible cases. 
\end{remark}

We present some simple claims regarding additive Gaussian channels which will be useful later.

\begin{claim}
\label{cl:cmg}
Consider the following vector additive Gaussian product channel with identical components
\begin{align*} 
\Yv_1 &= G \Xv_1 + \Zv_1\\
\Yv_2 &= G \Xv_2 + \Zv_2
\end{align*}
Further let $Z_1, Z_2$ be independent and distributed as $ \Nc(0,I)$.
Define 
$$\tilde{\Xv} = \frac{1}{\sqrt{2}} ( \Xv_1 + \Xv_2), \quad \Xv' = \frac{1}{\sqrt{2}} ( \Xv_1 - \Xv_2), \quad \tilde{\Yv} = \frac{1}{\sqrt{2}} ( \Yv_1 + \Yv_2), \quad \Yv' = \frac{1}{\sqrt{2}} ( \Yv_1 - \Yv_2).$$

Then $I(\Xv_1,\Xv_2;\Yv_1,\Yv_2) = I(\tilde{\Xv},\Xv';\tilde{\Yv},\Yv')$.
\end{claim}
\begin{proof}
The proof is a trivial consequence of the fact that
$h(\tilde{\Yv},\Yv') = h(\Yv_1,\Yv_2)$ and $h(\tilde{\Yv},\Yv'|\tilde{\Xv},\Xv') = h(\tilde{\Zv}, \Zv') = h(\Zv_1, \Zv_2)=h(\Yv_1,\Yv_2|\Xv_1, \Xv_2)$ where
$\tilde{\Zv} = \frac{1}{\sqrt{2}} ( \Zv_1 + \Zv_2), \Zv' = \frac{1}{\sqrt{2}} ( \Zv_1 - \Zv_2).$
\end{proof}
\begin{remark}
An interesting consequence of Gaussian noise is that $\tilde{\Zv}$ and $\Zv\rq{}$ are again independent and distributed according to $\Nc(0,I)$. Hence $\tilde{\Yv}, \Yv\rq{}$ can be regarded as the outputs of the Gaussian channel when the inputs are distributed according to $\tilde{\Xv}, \Xv\rq{}$. This observation is peculiar to additive Gaussian channels.
\end{remark}

\begin{claim}
\label{cl:cm1}
In vector additive Gaussian product broadcast channels with invertible channel gain matrices, the random variables $\Yv_{11}$ and $\Yv_{22}$ are independent {\em if and only if} $\Xv_1$ and $\Xv_2$ are independent.
\end{claim}
\begin{proof}
Here we prove the non-trivial direction. Suppose $\Yv_{11}$ and $\Yv_{22}$ are independent. We know that $\Yv_{11} = G_{11} \Xv_1 + \Zv_{11} $ and $\Yv_{22} = G_{22}\Xv_2 + \Zv_{22}$ where $\Zv_{11}, \Zv_{22}$ are mutually independent and independent of the pair $\Xv_1,\Xv_2$. Taking characteristic functions we see that
$$ \E\left(e^{i(\tv_1\cdot \Yv_{11} + \tv_2\cdot \Yv_{22})}\right) = \E\left(e^{i\tv_1\cdot \Yv_{11}}\right)  \E\left(e^{i\tv_2\cdot \Yv_{22}}\right) = \E\left(e^{i\tv_1 \cdot \Zv_{11}}\right) \E\left(e^{i\tv_1 \cdot G_{11}\Xv_1}\right)  \E\left(e^{i\tv_2 \cdot G_{22}\Xv_2}\right) \E\left(e^{i\tv_2 \cdot \Zv_{22}}\right). $$
On the other hand  
$$ \E\left(e^{i(\tv_1\cdot \Yv_{11} + \tv_2 \cdot \Yv_{22})}\right) = \E\left(e^{i\tv_1 \cdot \Zv_{11}}\right) \E\left(e^{i(\tv_1 \cdot G_{11}\Xv_1+\tv_2\cdot G_{22} \Xv_2)}\right)  \E\left(e^{i\tv_2 \cdot \Zv_{22}}\right).
$$ 
Since $  \E\left(e^{i\tv_1 \cdot \Zv_{11}}\right) ,  \E\left(e^{i\tv_2 \cdot \Zv_{22}}\right) > 0 ~ \forall \tv_1, \tv_2$ we have that
$$  \E\left(e^{i(\tv_1 \cdot G_{11}\Xv_1+\tv_2 \cdot G_{22} \Xv_2)}\right) = \E\left(e^{i\tv_1 \cdot G_{11}\Xv_1}\right)  \E\left(e^{i\tv_2 \cdot G_{22}\Xv_2}\right),  ~ \forall \tv_1, \tv_2.$$
Hence $G_{11} \Xv_1$ and $G_{22} \Xv_2$ are independent;  since $G_{11}$ and $G_{22}$ are invertible,  $\Xv_1$ and $\Xv_2$ are independent.
\end{proof}

\section{Optimality of Gaussian via factorization of concave envelopes}

We devise a new technique to show that Gaussian distribution achieves the maximum value of an optimization problem, subject to a covariance constraint. Though some of the results have been known earlier\cite{liv07}, the technique presented here allows us to obtain much broader results.

The main idea behind the approach is to show that if a certain  $X$ (say zero mean) achieves the maximum value of an optimization problem, then so does $\frac{1}{\sqrt{2}} (X_1 + X_2)$ and $\frac{1}{\sqrt{2}} (X_1 - X_2)$; where $X_1, X_2$ are two i.i.d. copies of $X$. Further we will show that $\frac{1}{\sqrt{2}} (X_1 + X_2)$ and $\frac{1}{\sqrt{2}} (X_1 - X_2)$ have to be independent as well, which forces the initial distribution to be Gaussian, see Theorem \ref{th:gho} and Corollary \ref{co:tworv} in Appendix \ref{sse:gau}. Alternately, one can repeat averaging procedure inductively and use central limit theorem to conclude that Gaussian distribution achieves the maximum. To show the first step we go to the two-letter version\footnote{A two letter version of a channel $\qmf(y|x)$ is a product channel consisting of identical components $\qmf(y_1|x_1) \times \qmf(y_2|x_2)$.} of the channel, use a {\em factorization property} of the function involved and then Claim \ref{cl:cmg} to move from the pair $X_1,X_2$ to  $\frac{1}{\sqrt{2}} (X_1 + X_2)$.

\begin{remark}
It is worth noting the remarkable similarity of the structure of the arguments that follow for the three optimization problems below for which we show the optimality of Gaussian. In particular the first example, though trivial, contains most of the key intuitive elements. 
\end{remark}

\subsection{Example 1: Mutual information}
\label{sse:fun-I}

Let $\Yv = G \Xv + \Zv$ represent a point-to-point channel, where $\Zv \sim \Nc(0,I)$ and $G$ is invertible. Given $K \psd 0$, consider the following optimization problem:
$$ \textrm{V}(K) = \max_{\Xv: \E(\Xv \Xv^T) \nsd K} I(\Xv;\Yv). $$

\begin{remark}
By writing $\max$ instead of $\sup$ we are indeed claiming the existence of a maximizing distribution. This is a non-trivial technical issue that we will deal with (in the Appendix) for the newer functions that we consider. The same arguments used for establishing Claim \ref{cl:max-S} can be used (essentially verbatim) to imply the existence of a maximizing distribution here. Furthermore for the above optimization problem it is well-known that $\Xv \sim \Nc(0,K)$ achieves $\textrm V(K)$ and the aim here is to give a simple illustration of our approach.
\end{remark}

Consider a product channel consisting of two identical components of the point-to-point channel described above: $\qmf(\Yv_{1}|\Xv_1)\times \qmf(\Yv_{2}|\Xv_2)$. We call the below claim as the {\em factorization property} of mutual information.

\begin{claim}
\label{cl:sep-I}
The following inequality holds for the product channel
$$ I(\Xv_1,\Xv_2;\Yv_1,\Yv_2) \leq  I(\Xv_1;\Yv_1) + I(\Xv_2;\Yv_2). $$
Further if equality is achieved at some $p(\xv_1,\xv_2)$ then $\Xv_1, \Xv_2$ must be independent.
\end{claim}
\begin{proof}
The proof is essentially a consequence of the following equality for product channels
$$ I(\Xv_1,\Xv_2;\Yv_1,\Yv_2) =  I(\Xv_1;\Yv_1) + I(\Xv_2;\Yv_2) - I(\Yv_1; \Yv_2). $$
Further, if equality holds then $\Yv_1, \Yv_2$ must be independent, which from Claim \ref{cl:cm1} implies that $\Xv_1$ and $\Xv_2$ are independent.
\end{proof}

Let $p_*(\xv)$ be a zero mean distribution that achieves $\textrm{V}(K)$.
\begin{claim}
\label{cl:twolet-I}
Let $(\Xv_1,\Xv_2) \sim p_*(\xv_1)p_*(\xv_2)$ be two i.i.d. copies of $p_*(\xv)$.  Then the following distributions $\tilde{\Xv} = \frac{1}{\sqrt{2}} \left( \Xv_{1} + \Xv_{2} \right)$, $\Xv'= \frac{1}{\sqrt{2}} \left( \Xv_{1} - \Xv_{2} \right)$  also achieve ${\rm V}(K)$. Further the random variables $\tilde{\Xv}, \Xv'$ are independent.
\end{claim}
\begin{proof}
Let $\tilde{\Yv} = \frac{1}{\sqrt{2}} \left( \Yv_{1} + \Yv_{2} \right)$, $\Yv'= \frac{1}{\sqrt{2}} \left( \Yv_{1} - \Yv_{2} \right)$.
The claim is a consequence of Claim \ref{cl:sep-I} and the following observations:
\begin{align*}
2 \textrm{V}(K) &=  I(\Xv_1;\Yv_1) + I(\Xv_2;\Yv_2)  \\
& = I(\Xv_1,\Xv_2;\Yv_1,\Yv_2) \\
& \stackrel{(a)}{=}  I(\tilde{\Xv},\Xv';\tilde{\Yv},\Yv') \\
& \stackrel{(b)}{\leq} I(\tilde{\Xv};\tilde{\Yv}) + I(\Xv';\Yv') \\
& \leq \textrm{V}(K) + \textrm{V}(K) =  2 \textrm{V}(K).
\end{align*}
Here the first equality comes because $p_*(\xv)$  achieves $\textrm{V}(K)$, the second one because $\Xv_1$ and $\Xv_2$ are independent. Equality $(a)$ is a consequence of Claim \ref{cl:cmg}, Inequality $(b)$ is a consequence of Claim \ref{cl:sep-I}, and the last inequality follows from the following:
$$ \E(\tilde{\Xv}\tilde{\Xv}^T) = \E(\Xv' \Xv'^T) = \frac 12 \left(\E(\Xv_1 \Xv_1^T) + \E(\Xv_2 \Xv_2^T)\right) \nsd K,$$
and the definition of \textrm{V}(K). Since the extremes match, all inequalities must be equalities. Hence $(b)$ must be an equality, which implies from Claim \ref{cl:sep-I} that $\tilde{\Xv}, \Xv'$ are independent. Similarly we require $I(\tilde{\Xv};\tilde{\Yv}) =  I(\Xv';\Yv') = \textrm{V}(K)$ as desired.
\end{proof}
Hence we have shown that $\Xv \sim p^*(\xv)$ that achieves a maximum has the following property: If $(\Xv_1,\Xv_2)$ are i.i.d. copies each distributed according to $p_*(\xv)$, then $\Xv_1 + \Xv_2$ and $\Xv_1 - \Xv_2$ are also independent. Thus from  Theorem \ref{th:gho} and Corollary \ref{co:tworv}(Appendix \ref{sse:gau}) we have that $\Xv \sim \Nc(0,K\rq{})$ for some $K\rq{} \nsd K$. 
Alternately, one could also use the following approach: For any $\Xv \sim p^*(\xv)$ (assume zero mean) that achieve the maximum, we know that $\frac{1}{\sqrt{2}} \left( \Xv_{1} + \Xv_{2} \right)$ also achieves the maximum. Hence proceeding by induction, we can use Central limit theorem to deduce that the Gaussian distribution also achieves the maximum. This alternate approach is elaborated for the next example in  Appendix \ref{se:altclt}. 
\begin{remark}
For this example, we can use the monotonicity of the $\log |\cdot|$ function to deduce that $K\rq{} = K$. In the examples that follow below we do not have any such monotonicity. Hence, we will only establish that the optimizing distribution is a Gaussian, which is sufficient for our purposes. 
\end{remark}

\subsection{Example 2: Difference of mutual informations}
\label{sse:fun-S}
Consider a vector additive Gaussian broadcast channel. For $\la > 1$ let the following function of $p(x)$ be defined by
$$ \mathsf{s}_\la(\Xv) := I(\Xv;\Yv_1) - \la I(\Xv;\Yv_2). $$
Let $ \mathsf{s}_\la(\Xv|V) :=   I(\Xv;\Yv_1|V) - \la I(\Xv;\Yv_2|V). $

Further define
$$ S_\la(\Xv) := \ce(\mathsf{s}_\la(\Xv)) $$
denote the {\em upper concave envelope}\footnote{The upper  concave envelope of a function $f(x)$ is the smallest concave function $g(x)$ such that $g(x) \geq f(x), \forall x$.} of $\mathsf{s}_\la(\Xv)$. It is a straightforward exercise to see that
$$ \ce(\mathsf{s}_\la(\Xv)) = \sup_{\substack{p(v|\xv): \\V \to \Xv \to (\Yv_1,\Yv_2)}} I(\Xv;\Yv_1|V) - \la I(\Xv;\Yv_2|V) = \sup_{p(v|\xv)}  \mathsf{s}_\la(\Xv|V). $$
We also define $S_\la(\Xv|V) := \sum_{v} p(v)  S_\la(\Xv|V=v)$ for finite $V$ and its natural extension for arbitrary $V$.

\begin{remark}
We will try to keep the language simple in the main body of this paper. In the Appendix we will  deal with the various technical issues with due diligence.
\end{remark}

For a product broadcast channel $\qmf_1(\yv_{11},\yv_{21}|\xv_1) \times \qmf_2(\yv_{12},\yv_{22}|\xv_2)$ let $S_\la(\Xv_1,\Xv_2)$ denote the corresponding upper concave envelope.
The following claim is referred to as the \lq\lq{}factorization of $S_\la(\Xv_1,\Xv_2)$\rq\rq{}.

\begin{claim}
\label{cl:sep-S}
The following inequality holds for product broadcast channels
$$ S_\la(\Xv_1,\Xv_2) \leq S_\la(\Xv_1| \Yv_{22}) + S_\la(\Xv_2| \Yv_{11}) \leq S_\la(\Xv_1) + S_\la(\Xv_2). $$
For additive Gaussian noise broadcast channels if $p(v|\xv_1,\xv_2)$  realizes $ S_\la(\Xv_1,\Xv_2)$, i.e.  $ S_\la(\Xv_1,\Xv_2) = \mathsf{s}_\la(\Xv_1,\Xv_2|V)$,  and equality is achieved above i.e.  $S_\la(\Xv_1,\Xv_2)=S_\la(\Xv_1) + S_\la(\Xv_2)$, then all of the following must be true
\begin{enumerate}
\item $\Xv_1$ and $\Xv_2$ are conditionally independent of $V$
\item $V,\Xv_1$ achieves $S_\la(\Xv_1)$
\item $V,\Xv_2$ achieves $S_\la(\Xv_2).$
\end{enumerate}
\end{claim}
\begin{proof}
For any $p(v|\xv_1,\xv_2)$ observe the following
{\small \begin{align*}
& I(\Xv_1, \Xv_2; \Yv_{11}, \Yv_{12}|V) - \la I(\Xv_1, \Xv_2; \Yv_{21}, \Yv_{22}|V) \\
& \quad = I(\Xv_1; \Yv_{11}|V) + I(\Xv_2; \Yv_{12}|V, \Yv_{11}) - \la I(\Xv_2; \Yv_{22}|V) - \la I(\Xv_1; \Yv_{21}|V, \Yv_{22}) \\
& \quad =  I(\Xv_1; \Yv_{11}|V, \Yv_{22}) + I(\Xv_2; \Yv_{12}|V, \Yv_{11}) - \la I(\Xv_2; \Yv_{22}|V, \Yv_{11}) - \la I(\Xv_1; \Yv_{21}|V, \Yv_{22}) - (\la - 1) I(\Yv_{11};\Yv_{22}|V) \\
& \quad \leq S_\la(\Xv_1|\Yv_{22}) + S_\la(\Xv_2|\Yv_{11}) - (\la - 1) I(\Yv_{11};\Yv_{22}|V)  \\
& \quad \leq S_\la(\Xv_1) + S_\la(\Xv_2) - (\la - 1) I(\Yv_{11};\Yv_{22}|V) \\
& \quad \leq  S_\la(\Xv_1) + S_\la(\Xv_2).
\end{align*}}\noindent 
Since equality holds all inequalities are tight. Hence $\Yv_1$ and $\Yv_2$ are conditionally independent of $V$ implying  that $\Xv_1$ and $\Xv_2$ are conditionally independent of $V$ (Claim \ref{cl:cmg}). Hence
{\small $$ I(\Xv_1; \Yv_{11}|V, \Yv_{22}) - \la I(\Xv_1; \Yv_{21}|V, \Yv_{22}) =  I(\Xv_1; \Yv_{11}|V) - \la I(\Xv_1; \Yv_{21}|V) =  S_\la(\Xv_1), $$
$$  I(\Xv_2; \Yv_{12}|V, \Yv_{11}) - \la I(\Xv_2; \Yv_{22}|V, \Yv_{11}) =  I(\Xv_2; \Yv_{12}|V) - \la I(\Xv_2; \Yv_{22}|V) = S_\la(\Xv_2). $$}\noindent 
This completes the proof.
\end{proof}

\subsubsection{Maximizing the concave envelope subject to a covariance constraint}

Consider an Additive Gaussian Noise broadcast channel $\qmf(y_1,y_2|x)$.
For $K \psd 0$, define
{\small $$ \textrm{V}_\la(K) = \sup_{\Xv:\E(\Xv\Xv^T)\nsd K} S_\la(\Xv). $$ }

\begin{claim}
\label{cl:max-S}
There is a pair of random variables $(V_*,\Xv_*)$ with $|V_*| \leq \frac{t(t+1)}{2} + 1$ such that
{\small $$ {\rm V}_\la(K) = \mathsf{s}_\la(\Xv_*|V_*). $$}
\end{claim}
\begin{proof}
This is a technical claim that shows that the supremum is indeed attained.  The details are present in the Appendix \ref{se:maxd}.
\end{proof}

The goal of this section is to show that a single Gaussian distribution achieves $\textrm{V}_\la(K)$, i.e. we can take $V$ to be trivial and $X \sim \Nc(0,K\rq{}), K\rq{} \nsd K$. (This result is known and was first shown by Liu and Vishwanath\cite{liv07} using perturbation based techniques. We use this here as a non-trivial illustration of our technique and then our final result in the next section is new.)

\medskip

Consider a product channel consisting of two identical components $\qmf(\Yv_{11},\Yv_{21}|\Xv_1)\times \qmf(\Yv_{12},\Yv_{22}|\Xv_2)$. 

\smallskip

\noindent{\em Notation}: In the remainder of the section we assume that $p_*(v,\xv)$ achieves $\textrm{V}_\la(K)$ ,  $|V| = m  \leq  \frac{t(t+1)}{2} + 1$ and $\Xv_v$ be a centered random variable (zero-mean) distributed according to $p(\Xv|V=v)$. Further let $K_v = \E(\Xv_v \Xv_v^T)$. Then we have $\sum_{v=1}^m p_*(v) K_v \nsd K$ and in particular that $K_v$\rq{}s are bounded.

\begin{claim}
\label{cl:twolet-S}
Let $(V_1,V_2,\Xv_1,\Xv_2) \sim p_*(v_1,\xv_1)p_*(v_2,\xv_2)$ be two i.i.d. copies of $p_*(v,\xv)$. We assume that $|V| \leq  \frac{t(t+1)}{2} + 1$.  
Let
$$  \Vt = (V_1, V_2), \quad  \tilde{\Xv}|\big(\Vt = (v_1, v_2)\big)  \sim  \frac{1}{\sqrt{2}} \left( \Xv_{v_1} + \Xv_{v_2} \right), \quad \Xv'|\big(\Vt = (v_1, v_2)\big)  \sim  \frac{1}{\sqrt{2}} \left( \Xv_{v_1} - \Xv_{v_2} \right).$$
In the above we take $\Xv_{v_1}$ and $\Xv_{v_2}$ to be independent random variables. Then the following hold:
\begin{enumerate}
\item $\tilde{\Xv}, \Xv\rq{}$ are conditionally independent given $\tilde{V}$.
\item $\Vt, \tilde{\Xv}$ achieves ${\rm V}_\la(K)$.
\item $\Vt, \Xv'$ achieves ${\rm V}_\la(K)$.
\end{enumerate}
\end{claim}
\begin{proof}
\begin{align*}
2 \textrm{V}_\la(K) &= \mathsf{s}_\la(\Xv_1|V_1) +  \mathsf{s}_\la(\Xv_2|V_2)  \\
& = \mathsf{s}_\la(\Xv_1,\Xv_2|V_1,V_2) \\
& \stackrel{(a)}{=}  \mathsf{s}_\la(\tilde{\Xv},\Xv\rq{}|\Vt) \\
&  \stackrel{(b)}{\leq} S_\la(\tilde{\Xv}, \Xv\rq{}) \\
& \stackrel{(c)}{\leq} S_\la(\tilde{\Xv}) + S_\la(\Xv\rq{}) \\
& \leq \textrm{V}_\la(K) + \textrm{V}_\la(K) =  2 \textrm{V}_\la(K).
\end{align*}
Here the first equality comes because $p_*(v,\xv)$  achieves $\textrm{V}_\la(K)$, the second one because $(V_1,\Xv_1)$ and $(V_2,\Xv_2)$ are independent. Equality $(a)$ is a consequence of Claim \ref{cl:cmg}, inequality $(c)$ is a consequence of Claim \ref{cl:sep-S}, and the last inequality follows from the following:
$$ \E(\tilde{\Xv}\tilde{\Xv}^T) = \E(\Xv' \Xv'^T) = \sum_{v_1,v_2} p_*(v_1)p_*(v_2) \frac{(K_{v_1} + K_{v_2})}{2}=  \sum_{v=1}^m p_*(v) K_v \nsd K,$$
and the definition of $\textrm{V}_\la(K)$. Since the extremes match, all inequalities must be equalities. Hence $(b)$ must be an equality, $p(\vt,\tilde{\xv},\xv\rq{})$ achieves $S_\la(\tilde{\Xv}, \Xv\rq{})$; and since $(c)$ is also equality from Claim \ref{cl:sep-S} we conclude that $\tilde{\Xv}, \Xv'$ are conditionally independent of $\Vt$.Furthermore, we also obtain that  $p(\vt|\tilde{\Xv})$ achieves $S_\la(\tilde{\Xv})$, which from the last inequality matches ${\rm V}_\la(K)$. Similarly for $p(\vt|\Xv\rq{})$.
\end{proof}

As a consequence, $\Xv_{v_1}$, $\Xv_{v_2}$ are independent random variables and $ \left( \Xv_{v_1} + \Xv_{v_2} \right),  \left( \Xv_{v_1} - \Xv_{v_2} \right)$ are also independent random variables. Thus from Corollary \ref{co:tworv} (in Appendix \ref{sse:gau}) $\Xv_{v_1}, \Xv_{v_2}$ are Gaussians, say having the same distribution as $\Xv_v\sim \Nc(0,K\rq{})$. Since $v_1, v_2$ are arbitrary, all $\Xv_{v_i}$ are  Gaussians, having the same distribution as $\Xv_v$. Then
$$ {\rm V}_\la(K) = \sum_{i=1}^m p_*(v_i) \mathsf{s}_\la(\Xv_{v_i}) = \sum_{i=1}^m p_*(v_i) \mathsf{s}_\la(\Xv_v) = \mathsf{s}_\la(\Xv_v). $$
Hence we obtain the following theorem.

\begin{theorem}
\label{th:lvalt}
There exists $\Xv_* \sim \Nc(0,K'), K' \nsd K$ such that ${\rm V}_\la(K) = \mathsf{s}_\la(\Xv_*).$
\end{theorem}

{\em Remark}: Notice that we never used the precise form of $S_\la(X)$ but just used that the implications of Claim \ref{cl:sep-S}. In the next section we will define a new concave envelope that will also satisfy a condition similar to Claim \ref{cl:sep-S}, and then establish the optimality of Gaussian.

\begin{corollary}
\label{co:tri}
If $\Xv \sim \Nc(0,K)$ then there exists $\Xv_* \sim \Nc(0,K'), K' \nsd K$ such that
$S_\la(\Xv) = \mathsf{s}_\la(\Xv_*) = {\rm V}_\la(K).$
\end{corollary}
\begin{proof}
Clearly from Theorem \ref{th:lvalt} and definition of $ {\rm V}_\la(K)$ we have 
$$ S_\la(\Xv) \leq {\rm V}_\la(K) = \mathsf{s}_\la(\Xv_*). $$
On the other hand let $\Xv' \sim \Nc(0, K-K')$ be independent of $\Xv_*$. Note that $\Xv \sim \Xv' + \Xv_*$ and
\[ S_\la(\Xv) =\sup_V \mathsf{s}_\la(\Xv|V) \geq \mathsf{s}_\la(\Xv|\Xv') = \mathsf{s}_\la(\Xv_*). \qedhere \]
\end{proof}

\subsection{Example 3: A more complicated example}
\label{sse:fun-T}
The function we considered in the previous section can be used determine the capacity region of vector Gaussian broadcast channel with only private messages\cite{liv07}(see Section \ref{sse:pm}). The function we consider in this section will enable us to determine the capacity region of vector Gaussian broadcast channel with common message as well (see Section \ref{sse:cm}).

 For $\la_0, \la_1, \la_2 > 0$ and for $\alpha \in [0,1]$ (and $\bar\alpha:=1-\alpha$) consider the following function of $p(\xv)$ defined by
$$ \mathsf{t}_\lav(\Xv) :=-\la_0\alpha I(\Xv;\Yv_1) - \la_0 \bar \alpha I(\Xv;\Yv_2) + (\la_1+\la_2) I(\Xv;\Yv_2) +\la_1 S_{\frac{\la_1 + \la_2}{\la_1}}(\Xv). $$

Further let 
$$ T_\lav(\Xv) := \ce(\mathsf{t}_\lav(\Xv)) $$
denote the upper concave envelope of $\mathsf{t}_\lav(\Xv)$. 
 It is easy to see that
{\small $$ \ce(\mathsf{t}_\lav(\Xv)) = \sup_{p(w|\xv)} -\la_0\alpha I(\Xv;\Yv_1|W) - \la_0 \bar \alpha I(\Xv;\Yv_2|W) + (\la_1+\la_2) I(\Xv;\Yv_2|W) + \la_1 S_{\frac{\la_1 + \la_2}{\la_1}}(\Xv|W). $$} 

For a product broadcast channel $\qmf_1(\yv_{11},\yv_{21}|\xv_1) \times \qmf_2(\yv_{12},\yv_{22}|\xv_2)$ let $T_\lav(\Xv_1,\Xv_2)$ denote the corresponding upper concave envelope.
The following claim is referred to as the \lq\lq{}factorization of $T_\lav(\Xv_1,\Xv_2)$\rq\rq{}.

\begin{claim}
\label{cl:sep-T}
When $\la_0 > \la_1 + \la_2$ the following inequality holds for product broadcast channels
$$ T_\lav(\Xv_1,\Xv_2)  \leq T_\lav(\Xv_1|\Yv_{22}) +  T_\lav(\Xv_2|\Yv_{11})  \leq T_\lav(\Xv_1) + T_\lav(\Xv_2). $$
For additive Gaussian noise broadcast channels if $p(w,\Xv_1,\Xv_2)$  realizes $ T_\lav(\Xv_1,\Xv_2)$ and equality is achieved above then all of the following must be true
\begin{enumerate}
\item $\Xv_1$ and $\Xv_2$ are conditionally independent of $W$
\item $W,\Xv_1$ achieves $T_\lav(\Xv_1)$
\item $W,\Xv_2$ achieves $T_\lav(\Xv_2).$
\end{enumerate}
\end{claim}
\begin{proof}
Observe the following
{\small \begin{align*}
& -\la_0\alpha I(\Xv_1,\Xv_2;\Yv_{11},\Yv_{12}|W) - \la_0 \bar \alpha I(\Xv_1,\Xv_2;\Yv_{21},\Yv_{22}|W) + (\la_1+\la_2) I(\Xv_1,\Xv_2;\Yv_{21},\Yv_{22}|W) \\
& \qquad + \la_1 S_{\frac{\la_1 + \la_2}{\la_1}}(\Xv_1, \Xv_2|W) \\
& \quad \leq -\la_0\alpha I(\Xv_1;\Yv_{11}|W)  -\la_0\alpha I(\Xv_2;\Yv_{12}|W,\Yv_{11})  - \la_0 \bar \alpha I(\Xv_2;\Yv_{22}|W) - \la_0 \bar \alpha I(\Xv_1;\Yv_{21}|W,\Yv_{22}) \\
& \qquad \qquad + (\la_1+\la_2) I(\Xv_2;\Yv_{22}|W) + (\la_1+\la_2) I(\Xv_1;\Yv_{21}|W,\Yv_{22}) +  \la_1 S_{\frac{\la_1 + \la_2}{\la_1}}(\Xv_1|W,\Yv_{22}) +  \la_1 S_{\frac{\la_1 + \la_2}{\la_1}}( \Xv_2|W, \Yv_{11}) \\
& \quad \leq -\la_0\alpha I(\Xv_1;\Yv_{11}|W, \Yv_{22})  -\la_0\alpha I(\Xv_2;\Yv_{12}|W,\Yv_{11})  - \la_0 \bar \alpha I(\Xv_2;\Yv_{22}|W, \Yv_{11}) - \la_0 \bar \alpha I(\Xv_1;\Yv_{21}|W,\Yv_{22})\\ 
& \qquad \qquad + (\la_1+\la_2) I(\Xv_2;\Yv_{22}|W, \Yv_{11})  + (\la_1+\la_2) I(\Xv_1;\Yv_{21}|W,\Yv_{22})  +  \la_1 S_{\frac{\la_1 + \la_2}{\la_1}}(\Xv_1|W,\Yv_{22})\\
& \qquad \qquad  +  \la_1 S_{\frac{\la_1 + \la_2}{\la_1}}( \Xv_2|W, \Yv_{11}) - (\la_0 - \la_1 - \la_2) I(\Yv_{11};\Yv_{22}|W) \\
& \quad \leq T_\lav(\Xv_1|\Yv_{22}) + T_\lav(\Xv_2|\Yv_{11}) - (\la_0 - \la_1 - \la_2) I(\Yv_{11};\Yv_{22}|W)  \\
& \quad \leq T_\lav(\Xv_1) + T_\lav(\Xv_2) - (\la_0 - \la_1 - \la_2) I(\Yv_{11};\Yv_{22}|W).
\end{align*}}\noindent 
Since equality holds, using Claim~\ref{cl:cmg} we have $\Xv_1$ and $\Xv_2$ are conditionally independent of $W$. Further using this and the equality observe that 
{\small
\begin{align*}  
& -\la_0\alpha I(\Xv_1;\Yv_{11}|W, \Yv_{22}) - \la_0 \bar \alpha I(\Xv_1;\Yv_{21}|W,\Yv_{22}) + (\la_1+\la_2) I(\Xv_1;\Yv_{21}|W,\Yv_{22}) +  \la_1 S_{\frac{\la_1 + \la_2}{\la_1}}(\Xv_1|W,\Yv_{22}) \\
& \quad = -\la_0\alpha I(\Xv_1;\Yv_{11}|W) - \la_0 \bar \alpha I(\Xv_1;\Yv_{21}|W) + (\la_1+\la_2) I(\Xv_1;\Yv_{21}|W) +  \la_1 S_{\frac{\la_1 + \la_2}{\la_1}}(\Xv_1|W) \\
& \quad = T_\lav(\Xv_1).
\end{align*}}\noindent 
Similarly for $\Xv_2$. This completes the proof.
\end{proof}
{\em Remark}: The above claim is the equivalent of Claim \ref{cl:sep-S}.

\medskip

For $K \psd 0$, define
$$ {\rm \hat V}_\lav(K) = \sup_{\Xv:E(\Xv\Xv^T)\nsd K} T_\lav(\Xv). $$

\begin{claim}
\label{cl:max-T}
There exists a pair $(W_*, \Xv_*)$ with $|W_*| \leq \frac{t(t+1)}{2} + 1$  such that ${\rm \hat V}_\lav(K) = \mathsf{t}_\lav(\Xv_*|W_*).$
\end{claim}

\noindent{\em Notation}: In the remainder of the section we assume that $p_*(w,\xv)$ achieves ${\rm \hat V}_\lav(K)$,  $|W| = m  \leq  \frac{t(t+1)}{2} + 1$ and $\Xv_w$ be a centered random variable (zero-mean) distributed according to $p(\Xv|W=w)$. Further let $K_w = \E(\Xv_w \Xv_w^T)$. Then we have $\sum_{w=1}^m p_*(w) K_w \nsd K$ and in particular that $K_w$\rq{}s are bounded.

\begin{claim}
\label{cl:twolet-T}
Let $(W_1,W_2,\Xv_1,\Xv_2) \sim p_*(w_1,\xv_1)p_*(w_2,\xv_2)$ be two i.i.d. copies of $p_*(w,x)$. We assume that $|W| \leq  \frac{t(t+1)}{2} + 1$.  
Let
$$ \Wt = (W_1, W_2), \quad  \tilde{\Xv}|\big(\Wt = (w_1, w_2)\big)  \sim  \frac{1}{\sqrt{2}} \left( \Xv_{w_1} + \Xv_{w_2} \right), \quad \Xv'|\big(\Wt = (w_1, w_2)\big)  \sim  \frac{1}{\sqrt{2}} \left( \Xv_{w_1} - \Xv_{w_2} \right).$$
In the above we take $\Xv_{w_1}$ and $\Xv_{w_2}$ to be independent random variables. Then the following hold:
\begin{enumerate}
\item $\tilde{\Xv}, \Xv\rq{}$ are conditionally independent given $\tilde{W}$.
\item  $\Wt, \tilde{\Xv}$ achieves ${\rm \hat V}_\lav(K)$.
\item  $\Wt, \Xv'$ achieves ${\rm \hat V}_\lav(K)$.
\end{enumerate}
\end{claim}
\begin{proof}
\begin{align*}
2 {\rm \hat V}_\lav(K) &= \mathsf{t}_\lav(\Xv_1|W_1) +  \mathsf{t}_\lav(\Xv_2|W_2)  \\
& = \mathsf{t}_\lav(\Xv_1,\Xv_2|W_1,W_2) \\
& \stackrel{(a)}{=}  \mathsf{t}_\lav(\tilde{\Xv},\Xv\rq{}|\Wt) \\
& \stackrel{(b)}{\leq} T_\lav(\tilde{\Xv}, \Xv\rq{}) \\
& \stackrel{(c)}{\leq} T_\lav(\tilde{\Xv}) + T_\lav(\Xv\rq{}) \\
& \leq {\rm \hat V}_\lav(K) + {\rm \hat V}_\lav(K) =  2 {\rm \hat V}_\lav(K).
\end{align*}
The proof mirrors that of Claim \ref{cl:twolet-S}.
Here the first equality comes because $p_*(w,\xv)$  achieves ${\rm \hat V}_\lav(K)$, the second one because $(W_1,\Xv_1)$ and $(W_2,\Xv_2)$ are independent. Equality $(a)$ is a consequence of Claim \ref{cl:cmg}, inequality $(c)$ is a consequence of Claim \ref{cl:sep-T}, and the last inequality follows from the following:
$$ \E(\tilde{\Xv}\tilde{\Xv}^T) = \E(\Xv' \Xv'^T) = \sum_{w_1,w_2} p_*(w_1)p_*(w_2) \frac{(K_{w_1} + K_{w_2})}{2}=  \sum_{w=1}^m p_*(w) K_w \nsd K,$$
and the definition of ${\rm \hat V}_\lav(K)$. Since the extremes match, all inequalities must be equalities. Hence $(b)$ must be an equality, $p(\wt,\tilde{\xv},\xv\rq{})$ achieves $T_\lav(\tilde{\Xv}, \Xv\rq{})$; and since $(c)$ is also equality from Claim \ref{cl:sep-T} we conclude that $\tilde{\Xv}, \Xv'$ are conditionally independent of $\Wt$. Furthermore, we also obtain that  $p(\wt|\tilde{\Xv})$ achieves $T_\lav(\tilde{\Xv})$, which from the last inequality matches ${\rm \hat V}_\lav(K)$. Similarly for $p(\wt|\Xv\rq{})$.
\end{proof}

As a consequence, $\Xv_{w_1}$, $\Xv_{w_2}$ are independent random variables and $ \left( \Xv_{w_1} + \Xv_{w_2} \right),  \left( \Xv_{w_1} - \Xv_{w_2} \right)$ are also independent random variables. Thus from Corollary \ref{co:tworv} (in Appendix \ref{sse:gau}) $\Xv_{w_1}, \Xv_{w_2}$ are Gaussians, say having the same distribution as $\Xv_w\sim \Nc(0,K\rq{})$. Since $w_1, w_2$ are arbitrary, all $\Xv_{w_i}$ are  Gaussians, having the same distribution as $\Xv_w$. Then
$$ {\rm \hat V}_\lav(K) = \sum_{i=1}^m p_*(w_i) \mathsf{t}_\lav(\Xv_{w_i}) = \sum_{i=1}^m p_*(w_i) \mathsf{t}_\lav(\Xv_w) = \mathsf{t}_\lav(\Xv_w). $$
Hence we obtain the following theorem.

\begin{theorem}
\label{th:gnnew}
There exists $\Xv_* \sim \Nc(0,K'), K' \nsd K$ such that ${\rm \hat V}_\lav(K) = \mathsf{t}_\lav(\Xv_*).$
\end{theorem}

\begin{corollary}
\label{co:tri2}
If $\Xv \sim \Nc(0,K)$ then there exists $\Xv_{1*} \sim \Nc(0,K_1)$ and an independent random variable $\Xv_{2*} \sim \Nc(0,K_2), K_1 + K_2=K\rq{} \nsd K$ such that
$T_\lav(\Xv) = \mathsf{t}_\lav(\Xv_{1*} + \Xv_{2*} ) = {\rm \hat V}_\lav(K)$ and $S_{\frac{\la_1 + \la_2}{\la_1}}(\Xv_{1*} + \Xv_{2*} ) = \mathsf{s}_{\frac{\la_1 + \la_2}{\la_1}}(\Xv_{1*}) = {\rm V}_{\frac{\la_1 + \la_2}{\la_1}}(K_1 + K_2)$.
\end{corollary}
\begin{proof}
Clearly from Theorem \ref{th:gnnew} and definition of $ {\rm \hat V}_\lav(K)$ we have 
$$ T_\lav(\Xv) \leq {\rm \hat V}_\lav(K) = \mathsf{t}_\lav(\Xv_*). $$
On the other hand let $\Xv' \sim \Nc(0, K-K')$ be independent of $\Xv_*$. Note that $\Xv \sim \Xv' + \Xv_*$ and
\[ T_\lav(\Xv) =\sup_W \mathsf{t}_\lav(\Xv|W) \geq \mathsf{t}_\lav(\Xv|\Xv') = \mathsf{t}_\lav(\Xv_*). \]
Now splitting of $\Xv_*$ into $\Xv_{1*}$, $\Xv_{2*}$ is possible by Corollary \ref{co:tri}.
\end{proof}

\section{Two capacity regions}

\subsection{Vector Gaussian Broadcast channel with private messages}
\label{sse:pm}

Consider a vector Gaussian broadcast channel with only private message requirements.
Let $\Cc$ be the capacity region. For $\la > 1$ we will seek to maximize the following expression
$$ \max_{(R_1, R_2) \in \Cc} R_1 + \la R_2. $$
The case for $\la < 1$ is dealt with similarly (with roles of $(Y_1,Y_2)$ interchanged). The case for $\la=1$ follows by continuity.

Here we consider the Korner-Marton outer bound and Marton\rq{}s inner bound (both from \cite{mar79}) to the capacity region of the broadcast channel. 
\begin{bound}
The union of rate pairs $(R_1, R_2)$ satisfying
\begin{align*}
R_2 &\leq I(V;Y_2) \\
R_1  &\leq I(X;Y_1) \\
R_1 + R_2 & \leq I(V;Y_2) + I(X;Y_1|V)
\end{align*}
over all $V \to X \to (Y_1,Y_2)$ forms an outer bound to the broadcast channel.
\end{bound}
Denote this region as $\Oc$.

\begin{bound}
The union of rate pairs $(R_1, R_2)$ satisfying
\begin{align*}
R_2 &\leq I(V;Y_2) \\
R_1  &\leq I(U;Y_1) \\
R_1 + R_2 & \leq I(U;Y_1) + I(V;Y_2) - I(U;V)
\end{align*}
over all $(U,V) \to X \to (Y_1,Y_2)$ forms an inner bound to the broadcast channel.
\end{bound}
Denote this region as $\Ic$.

One can adapt these inner and outer bounds to additive Gaussian setting by introducing a power constraint, i.e. an upper bound on  the trace of the covariance matrix, $\textrm{tr}(K)$.
However let us put a covariance constraint on $\Xv$ and denote $\Ic_K, \Cc_K, \Oc_K$ to be the corresponding inner bound, capacity region, and the outer bound.

Clearly we have
$$  \max_{(R_1, R_2) \in \Ic_K} R_1 + \la R_2 \leq  \max_{(R_1, R_2) \in \Cc_K} R_1 + \la R_2 \leq  \max_{(R_1, R_2) \in \Oc_K} R_1 + \la R_2. $$
To exhibit  the capacity region we will show that 
$$ \max_{(R_1, R_2) \in \Oc_K} R_1 + \la R_2  \leq  \max_{(R_1, R_2) \in \Ic_K} R_1 + \la R_2. $$
Thus Marton\rq{}s inner bound and Korner-Marton\rq{}s outer bound will match in this setting, and therefore also with the usual trace constraint.

Observe that
\begin{align*}
 \max_{(R_1, R_2) \in \Oc_K} R_1 + \la R_2  
 & \leq \sup_{\substack{V \to \Xv \to (\Yv_1,\Yv_2)\\ \E(\Xv \Xv^T) \nsd K}} \la I(V;\Yv_2) + I(\Xv;\Yv_1|V) \\
 & = \sup_{\substack{V \to \Xv \to (\Yv_1,\Yv_2)\\ \E(\Xv \Xv^T) \nsd K}} \la I(\Xv;\Yv_2) + I(\Xv;\Yv_1|V) - \la I(\Xv;\Yv_2|V) \\
 & \leq  \max_{\Xv: \E(\Xv \Xv^T) \nsd K}  \la I(\Xv;\Yv_2) +  \sup_{\substack{V \to \Xv \to (\Yv_1,\Yv_2)\\ \E(\Xv \Xv^T) \nsd K}}  I(\Xv;\Yv_1|V) - \la I(\Xv;\Yv_2|V) \\
 & \leq  \max_{\Xv: \E(\Xv \Xv^T) \nsd K}  \la I(\Xv;\Yv_2) + {\rm V}_\la(K).
\end{align*}

We know that the first term is maximized (Section \ref{sse:fun-S}) when $\Xv \sim \Nc(0,K)$ and ${\rm V}_\la(K)$ is achieved by  $\mathsf{s}_\la(\Xv_*)$ where $\Xv_* \sim \Nc(0,K\rq{}), K\rq{} \nsd K$.
Now let $V_* \sim \Nc(0,K-K\rq{})$ be independent of $\Xv_*$ and let $\Xv=V_* + \Xv_*$. Observe that this choice attains both maxima simultaneously.
Hence
$$  \max_{(R_1, R_2) \in \Oc_K} R_1 + \la R_2 \leq \la I(V_*;\Yv_2) + I(\Xv;\Yv_1|V_*) =  \la I(V_*;\Yv_2) + I(\Xv_*;\Yv_1|V_*). $$

\begin{lemma}[Dirty paper coding]
Let $\Xv = V_* + \Xv_*$ and $V_*$, $\Xv_*$ be independent Gaussians with covariances $K-K\rq{},K\rq{}$ respectively. Then there exists $U_*$ jointly Gaussian with $V_*$ such that
$$ I(\Xv;\Yv_1|V_*) = I(U_*; \Yv_1) - I(U_*;V_*). $$
Here $\Yv_1 = G \Xv + \Zv$, where $\Zv \sim \Nc(0,I)$ is independent of $V_*, \Xv_*$.
\label{le:dpc}
\end{lemma}
\begin{proof}
This well-known identification stems from the celebrated  paper\cite{cos83}. Set $U_* = \Xv_* + AV_*$ where
$ A = K\rq{} G^T(GK\rq{}G^T + I)^{-1}$ and this works (see Chapter 9.5 of \cite{elk12}).
\end{proof}

Now using $U_*$ as in the above lemma, we obtain
\begin{align*}
 \max_{(R_1, R_2) \in \Oc_K} R_1 + \la R_2 & \leq \la I(V_*;\Yv_2) + I(\Xv_*;\Yv_1|V_*) \\
 & =  \la I(V_*;\Yv_2) +   I(U_*; \Yv_1) - I(U_*;V_*).
 \end{align*}
 
However using Marton\rq{}s inner bound any rate pair satisfying $R_2 = I(V;\Yv_2)$, $R_1 = I(U;\Yv_1) - I(U;V)$ such that $\E(\Xv\Xv^T)\nsd K$ belongs to $\Ic_K$. Hence 
$$  \max_{(R_1, R_2) \in \Oc_K} R_1 + \la R_2 \leq  \la I(V_*;\Yv_2) +   I(U_*; \Yv_1) - I(U_*;V_*) \leq \max_{(R_1, R_2) \in \Ic_K} R_1 + \la R_2. $$
Thus the inner and outer bound match for vector Gaussian product channels establishing its capacity region.

\subsection{Vector Gaussian Broadcast channel with common message}
\label{sse:cm}
Consider a vector Gaussian broadcast channel with common and  private message requirements.
Let $\Cc$ be the capacity region. Assume $\la_0 > \la_1 + \la_2$. We will seek to maximize the following expression
$$ \max_{(R_0,R_1, R_2) \in \Cc} \la_0 R_0 + \la_1 R_1 + (\la_1 + \la_2) R_2. $$
\begin{remark}
The case of maximizing $ \la_0 R_0 + (\la_1 + \la_2) R_1 + \la_2 R_2 $ can be dealt with similarly. On the other hand if $\la_0 \leq (\la_1 + \la_2)$ then it suffices to consider the private messages capacity region. Actually the setting $\la_0 \geq 2 \la_1 + \la_2$ can be deduced from the degraded message sets capacity region and this is also known; however this will be subsumed in our treatment. Hence the setting we are considering is the only interesting unestablished case.
\end{remark}
  
In this section we consider the UVW outer bound\cite{nai11} and Marton\rq{}s inner bound\cite{mar79} to the capacity region of the broadcast channel with private and common messages.
\begin{bound} [UVW outer bound]
The union of rate triples $(R_0,R_1, R_2)$ satisfying
{\small \begin{align*}
R_0 &\leq \min \{I(W;Y_1), I(W;Y_2) \} \\
R_0 + R_1 &\leq  \min \{I(W;Y_1), I(W;Y_2) \} + I(U;Y_1|W) \\
R_0 + R_2 &\leq   \min \{I(W;Y_1), I(W;Y_2) \} + I(V;Y_2|W) \\
R_0 + R_1 + R_2 &\leq  \min \{I(W;Y_1), I(W;Y_2) \} + I(V;Y_2|W) + I(X;Y_1|V,W)\\ 
R_0 + R_1 + R_2 &\leq  \min \{I(W;Y_1), I(W;Y_2) \} + I(U;Y_1|W) + I(X;Y_2|U,W)
\end{align*}}\noindent 
over all $(U,V,W) \to X \to (Y_1,Y_2)$ forms an outer bound to the broadcast channel.
\end{bound}
As before, denote this region as $\Oc$.

\begin{bound}[Marton\rq{}s inner bound]
The union of rate pairs $(R_1, R_2)$ satisfying
{\small \begin{align*}
R_0 &\leq  \min \{I(W;Y_1), I(W;Y_2) \} \\
R_0 + R_1 &\leq I(U,W;Y_1) \\
R_0 + R_2  &\leq I(V,W;Y_2) \\
R_0 + R_1 + R_2 & \leq  \min \{I(W;Y_1), I(W;Y_2) \} + I(U;Y_1|W) + I(V;Y_2|W) - I(U;V|W)
\end{align*}}\noindent 
over all $(U,V) \to X \to (Y_1,Y_2)$ forms an inner bound to the broadcast channel.
\end{bound}
Denote this region as $\Ic$.
 
 \medskip
 
Impose a covariance constraint $K$ on $\Xv$ and denote $\Ic_K, \Cc_K, \Oc_K$ to be the corresponding inner bound, capacity region, and the outer bound respectively. Trivially we have
{\small  \begin{align*} \max_{(R_0,R_1, R_2) \in \Ic_K} \la_0 R_0 + \la_1 R_1 + (\la_1 + \la_2) R_2 &\leq  \max_{(R_0,R_1, R_2) \in \Cc_K} \la_0 R_0 + \la_1 R_1 + (\la_1 + \la_2) R_2 \\
& \leq  \max_{(R_0,R_1, R_2) \in \Oc_K} \la_0 R_0 + \la_1 R_1 + (\la_1 + \la_2) R_2. \end{align*}}

For any $\alpha \in [0,1]$ observe that (from first, third, and fourth constraints of UVW outer bound)
{\small \begin{align*}
&\max_{(R_0,R_1, R_2) \in \Oc_K} \la_0 R_0 + \la_1 R_1 + (\la_1 + \la_2) R_2 \\
& \quad \leq  \sup_{\substack{(V,W) \to \Xv \to (\Yv_1,\Yv_2)\\ \E(\Xv \Xv^T) \nsd K}} \alpha \la_0 I(W;\Yv_1) + \bar \alpha \la_0 I(W;\Yv_2) + (\la_1 + \la_2) I(V;\Yv_2|W) + \la_1 I(\Xv;\Yv_1|V,W)\\
& \quad =  \sup_{\substack{(V,W) \to \Xv \to (\Yv_1,\Yv_2)\\ \E(\Xv \Xv^T) \nsd K}}  \alpha \la_0 I(X;\Yv_1) + \bar \alpha \la_0 I(X;\Yv_2) -  \alpha \la_0 I(\Xv;\Yv_1|W) - \bar \alpha \la_0 I(\Xv;\Yv_2|W) \\
& \qquad \qquad  \quad \qquad \qquad +   (\la_1 + \la_2) I(\Xv;\Yv_2|W) + \la_1 I(\Xv;\Yv_1|V,W) -  (\la_1 + \la_2) I(\Xv;\Yv_2|V,W) \\
& \quad \leq \sup_{\substack{W \to \Xv \to (\Yv_1,\Yv_2)\\ \E(\Xv \Xv^T) \nsd K}}  \alpha \la_0 I(X;\Yv_1) + \bar \alpha \la_0 I(X;\Yv_2) -  \alpha \la_0 I(\Xv;\Yv_1|W) - \bar \alpha \la_0 I(\Xv;\Yv_2|W) \\
& \qquad \qquad \quad  \qquad \qquad +   (\la_1 + \la_2) I(\Xv;\Yv_2|W) +  \la_1 S_{\frac{\la_1 + \la_2}{\la_1}}(\Xv|W) \\
&  \quad \leq \max_{ \E(\Xv \Xv^T) \nsd K} \left( \alpha \la_0 I(X;\Yv_1) + \bar \alpha \la_0 I(X;\Yv_2)\right) +  \max_{\substack{W \to \Xv \to (\Yv_1,\Yv_2)\\ \E(\Xv \Xv^T) \nsd K}}\mathsf{t}_{\lav}(\Xv|W) \\
&  \quad \leq \max_{ \E(\Xv \Xv^T) \nsd K}  \left(\alpha \la_0 I(X;\Yv_1) + \bar \alpha \la_0 I(X;\Yv_2)\right)  + {\rm \hat V}_\lav(K).
\end{align*}}

We know that the first term is maximized (Section \ref{sse:fun-T}) when $\Xv \sim \Nc(0,K)$ and ${\rm\hat V}_\lav(K)$ is achieved by  $\mathsf{t}_\lav(\Xv_{1*}+\Xv_{2*})$ where $\Xv_{1*}, \Xv_{2*}$ are independent and $\Xv_{1*} \sim \Nc(0,K_1), \Xv_{2*} \sim \Nc(0,K_2), K_1 + K_2 \nsd K$, and $S_{\frac{\la_1 + \la_2}{\la_1}}(\Xv_{1*} + \Xv_{2*}) = \mathsf{s}_{\frac{\la_1 + \la_2}{\la_1}}(\Xv_{1*})$.  See Theorem \ref{th:gnnew} and Corollary \ref{co:tri2}.
Now let $W_* \sim \Nc(0,K - (K_1 + K_2))$ be independent of $\Xv_{1*},\Xv_{2*}$ and let $\Xv=W_* + \Xv_{1*} + \Xv_{2*}$. Observe that this choice attains both maxima simultaneously. For conforming to more standard notation, let us call $V_* = \Xv_{2*}$, thus $\Xv=W_* + \Xv_{1*} + V_*$. Thus
{\small \begin{align*}
&\max_{(R_0,R_1, R_2) \in \Oc_K} \la_0 R_0 + \la_1 R_1 + (\la_1 + \la_2) R_2 \\
& \quad \leq \alpha \la_0 I(X;\Yv_1) + \bar \alpha \la_0 I(X;\Yv_2) -  \alpha \la_0 I(\Xv;\Yv_1|W_*) - \bar \alpha \la_0 I(\Xv;\Yv_2|W_*) \\
& \qquad \qquad  \quad \qquad \qquad +   (\la_1 + \la_2) I(\Xv;\Yv_2|W_*) + \la_1 I(\Xv;\Yv_1|V_*,W_*) -  (\la_1 + \la_2) I(\Xv;\Yv_2|V_*,W_*) \\
& \quad =   \alpha\la_0 I(W_*;\Yv_1) + \bar \alpha \la_0 I(W_*;\Yv_2) + (\la_1 + \la_2) I(V_*;\Yv_2|W_*) + \la_1 I(\Xv;\Yv_1|V_*,W_*)\\
& \quad =   \alpha\la_0 I(W_*;\Yv_1) + \bar \alpha \la_0 I(W_*;\Yv_2) + (\la_1 + \la_2) I(V_*;\Yv_2|W_*) + \la_1 I(\Xv_{1*};\Yv_1|V_*,W_*)
\end{align*}}

Now using Lemma \ref{le:dpc} choose $U_* = X_{1*} + \tilde{A}V_*$ as before to have
$$ I(\Xv_{1*};\Yv_1|V_*,W_*) = I(U_*; \Yv_1|W_*) - I(U_*;V_*|W_*). $$

Hence
{\small \begin{align*}
&\max_{(R_0,R_1, R_2) \in \Oc_K} \la_0 R_0 + \la_1 R_1 + (\la_1 + \la_2) R_2 \\
& \leq \alpha\la_0 I(W_*;\Yv_1) + \bar \alpha \la_0 I(W_*;\Yv_2) + (\la_1 + \la_2) I(V_*;\Yv_2|W_*) + \la_1( I(U_*; \Yv_1|W_*) - I(U_*;V_*|W_*)) \\
& \leq \sup_{\substack{(U,V,W) \to \Xv \to (\Yv_1,\Yv_2)\\ \E(\Xv \Xv^T) \nsd K}}  \alpha\la_0 I(W;\Yv_1) + \bar \alpha \la_0 I(W;\Yv_2) + (\la_1 + \la_2) I(V;\Yv_2|W) + \la_1( I(U; \Yv_1|W) - I(U;V|W))
\end{align*}}\noindent 
Since the above holds for all $\alpha\in[0,1]$, we have
{\small \begin{align*}
&\max_{(R_0,R_1, R_2) \in \Oc_K} \la_0 R_0 + \la_1 R_1 + (\la_1 + \la_2) R_2 \\
&\leq \min_{\alpha \in [0,1]}  \sup_{\substack{(U,V,W) \to \Xv \to (\Yv_1,\Yv_2)\\ \E(\Xv \Xv^T) \nsd K}}  \alpha\la_0 I(W;\Yv_1) + \bar \alpha \la_0 I(W;\Yv_2) + (\la_1 + \la_2) I(V;\Yv_2|W) + \la_1 I(U; \Yv_{1}|W) - \la_1 I(U;V|W).
\end{align*}}

To complete the proof that the inner and outer bounds match we present the following Claim \ref{cl:mmin} (essentially established in \cite{ggny11b}). We will defer the proof of this claim to the Appendix \ref{sse:mm}.
 \begin{claim}
 \label{cl:mmin}
 We claim that
 {\small \begin{align*}
   & \min_{\alpha \in [0,1]}  \sup_{\substack{(U,V,W) \to \Xv \to (\Yv_1,\Yv_2)\\ \E(\Xv \Xv^T) \nsd K}}  \alpha\la_0 I(W;\Yv_1) + \bar \alpha \la_0 I(W;\Yv_2) + (\la_1 + \la_2) I(V;\Yv_2|W)  + \la_1 I(U; \Yv_{1}|W) - \la_1 I(U;V|W) \\
   & =  \sup_{\substack{(U,V,W) \to \Xv \to (\Yv_1,\Yv_2)\\ \E(\Xv \Xv^T) \nsd K}}  \min_{\alpha \in [0,1]}  \alpha\la_0 I(W;\Yv_1) + \bar \alpha \la_0 I(W;\Yv_2) + (\la_1 + \la_2) I(V;\Yv_2|W)  + \la_1 I(U; \Yv_{1}|W) - \la_1 I(U;V|W) \\
   & = \sup_{\substack{(U,V,W) \to \Xv \to (\Yv_1,\Yv_2)\\ \E(\Xv \Xv^T) \nsd K}} \la_0 \min\{ I(W;\Yv_1), I(W;\Yv_2) \} +  (\la_1 + \la_2) I(V;\Yv_2|W) + \la_1 I(U; \Yv_{1}|W) - \la_1 I(U;V|W).
   \end{align*}}
 \end{claim}
 
Now using Marton\rq{}s inner bound we can always achieve the following triples: $R_0 =  \min\{ I(W;\Yv_1), I(W;\Yv_2) \}$, $R_2 = I(V;\Yv_2|W)$, $R_1 =  I(U; \Yv_{1}|W) -  I(U;V|W).$ 
Hence
 {\small \begin{align*}
 &\max_{(R_0,R_1, R_2) \in \Oc_K} \la_0 R_0 + \la_1 R_1 + (\la_1 + \la_2) R_2 \\
 & \leq  \sup_{\substack{(U,V,W) \to \Xv \to (\Yv_1,\Yv_2)\\ \E(\Xv \Xv^T) \nsd K}} \la_0 \min\{ I(W;\Yv_1), I(W;\Yv_2) \} +  (\la_1 + \la_2) I(V;\Yv_2|W) + \la_1 I(U; \Yv_{1}|W) - \la_1 I(U;V|W) \\
 & \leq \max_{(R_0,R_1, R_2) \in \Ic_K} \la_0 R_0 + \la_1 R_1 + (\la_1 + \la_2) R_2.
 \end{align*}}
 
Hence Marton\rq{}s inner bound and UVW outer bound match and further the boundary is achieved via Gaussian signaling. To get a explicit characterization of the Gaussian signaling region (established as capacity here) please see the region given by equations $(2)-(4)$ in \cite{wss06a}.
 
\section{Conclusion}
We developed a new method to show the optimality of Gaussian distributions. We illustrated this technique for three examples and computed the capacity region of the two-receiver vector Gaussian broadcast channel with common and private messages. We can see several other problems where this technique can have immediate impact. Some of the mathematical tools and results in the Appendix can also be of independent interest.

\section*{Acknowledgement}
A lot of this work was motivated by the work on the discrete memoryless broadcast channel; a lot of which was jointly developed with Amin Gohari. The authors are also grateful to Venkat Anantharam, Abbas El Gamal, Amin Gohari, and Young-Han Kim  for their comments on early drafts and suggestions on improving the presentation. 

The work of Chandra Nair was partially supported by the following grants from the University Grants Committee of the Hong Kong Special Administrative Region, China: a) (Project No. AoE/E-02/08), b) GRF Project 415810. He also acknowledges the support from the Institute of Theoretical Computer Science and Communications (ITCSC) at the Chinese University of Hong Kong.

\bibliographystyle{amsplain}
\bibliography{mybiblio}

\appendix

\section{Some known results}

\subsection{A characterization of Gaussian distribution}
\label{sse:gau}
\begin{theorem}[Theorem 1 in \cite{gho62}]
\label{th:gho}
Let $\Xv_1,..,\Xv_n$ be $n$ mutually independent $t$-dimensional random column vectors, and let $A_1,..,A_n$ and $B_1,...,B_n$ be non-singular $t \times t$ matrices. If $\sum_{i=1}^n A_i \Xv_i$ is independent of $\sum_{i=1}^n B_i \Xv_i$, then the $\Xv_i$ are normally distributed.
\end{theorem}
\begin{remark}
In this paper we only use $A_i, B_i$ as multiples of $I$. In this case, the theorem follows from an earlier result of Skitovic. There were scalar versions of this known since the 30s, including Bernstein's theorem. The proof relies on solving the functional equations satisfied by the characteristic functions.
\end{remark}

\begin{corollary}
\label{co:tworv}
If $\Xv_1$ and $\Xv_2$ are zero-mean independent $t$-dimensional random column vectors, and if $\Xv_1 + \Xv_2$ and $\Xv_1 - \Xv_2$ are independent then $\Xv_1, \Xv_2$ are normally distributed with identical covariances.
\end{corollary}
\begin{proof}
The fact that $\Xv_1, \Xv_2$ are normally distributed follows from Theorem \ref{th:gho}. Now observe that $\E((\Xv_1 + \Xv_2)(\Xv_1 - \Xv_2)^T) = \E(\Xv_1 + \Xv_2) \E(\Xv_1 - \Xv_2)^T = \mathbf{0}$. On the other hand
\[ \E((\Xv_1 + \Xv_2)(\Xv_1 - \Xv_2)^T)  = \E(\Xv_1 \Xv_1^T) - \E(\Xv_2 \Xv_2^T).\qedhere \]
\end{proof}

\subsection{Min-max theorem}
\label{sse:mm}
We reproduce the following Corollary from the Appendix of \cite{ggny11b} (full version can be found in arXiv).

\begin{corollary}[Corollary 2 in arXiv version of \cite{ggny11b}]
\label{coro:mm} Let $\Lambda_d$ be the $d$-dimensional simplex, i.e. $\alpha_i\geq0$ and $\sum_{i=1}^d\alpha_i=1$. Let $\Pc$ be a set of probability distributions $p(u)$. Let $T_i(p(u)), i=1,..,d$ be a set of functions such that the set $\Ac$, defined by
\begin{align*}\Ac &=\{(a_1,a_2,...,a_d)\in \mathbb{R}^d: a_i\leq T_i(p(u))\mbox{ for some }~p(u) \in \Pc\},
\end{align*}
is a convex set.

Then
$$\sup_{p(u) \in \Pc} \min_{\alpha \in \Lambda_d} \sum_{i=1}^d \alpha_i T_i(p(u))  = \min_{\alpha \in \Lambda_d}\sup_{p(u)\in \Pc}  \sum_{i=1}^d \alpha_i T_i(p(u)). $$
\end{corollary}

We will now show how one can use the Corollary \ref{coro:mm} to establish Claim \ref{cl:mmin}.
\subsubsection*{ Proof of Claim \ref{cl:mmin}}
\begin{proof}
We take $\Pc$ as the set of $p(u,v,w,\xv)$  that satisfy the covariance constraint.
Here we take $d=2$ and set
\begin{align*}
T_1(p(u,v,w,x)) &= \la_0I(W;\Yv_1) + \la_1I(U;\Yv_1|W) +  (\la_1 + \la_2)I(V;\Yv_2|W) - \la_1I(U;V|W) \\
T_2(p(u,v,w,x)) &= \la_0I(W;\Yv_2) + \la_1I(U;\Yv_1|W) +  (\la_1 + \la_2)I(V;\Yv_2|W) - \la_1I(U;V|W)
\end{align*}
It is clear that the set
$$ \Ac = \{(a_1, a_2): a_1 \leq T_1(p(u,v,w,\xv)), a_2 \leq T_2(p(u,v,w,\xv)) \} $$
is a convex set. (In the standard manner, choose $\Wt = (W,Q)$, and when $Q=0$ choose $(U,V,W,\Xv) \sim p_1(u,v,w,\xv)$ and $Q=1$ choose $(U,V,W,\Xv) \sim p_2(u,v,w,\xv)$). Hence from Corollary \ref{coro:mm}, we have
{\small \begin{align*}
   & \min_{\alpha \in [0,1]}  \sup_{\substack{(U,V,W) \to \Xv \to (\Yv_1,\Yv_2)\\ \E(\Xv \Xv^T) \nsd K}}  \alpha\la_0 I(W;\Yv_1) + \bar \alpha \la_0 I(W;\Yv_2) + (\la_1 + \la_2) I(V;\Yv_2|W) + \la_1 I(U; \Yv_{1}|W) - \la_1 I(U;V|W) \\
   & =  \sup_{\substack{(U,V,W) \to \Xv \to (\Yv_1,\Yv_2)\\ \E(\Xv \Xv^T) \nsd K}}  \min_{\alpha \in [0,1]}  \alpha\la_0 I(W;\Yv_1) + \bar \alpha \la_0 I(W;\Yv_2) + (\la_1 + \la_2) I(V;\Yv_2|W)  + \la_1 I(U; \Yv_{1}|W) - \la_1 I(U;V|W) \\
   & = \sup_{\substack{(U,V,W) \to \Xv \to (\Yv_1,\Yv_2)\\ \E(\Xv \Xv^T) \nsd K}} \la_0 \min\{ I(W;\Yv_1), I(W;\Yv_2) \} +  (\la_1 + \la_2) I(V;\Yv_2|W) + \la_1 I(U; \Yv_{1}|W) - \la_1 I(U;V|W). \qedhere
   \end{align*}}
\end{proof}

\section{Existence of maximizing distributions}
\label{se:maxd}
The aim of this section is to give formal proofs of Claims \ref{cl:max-S} and \ref{cl:max-T} as our arguments critically hinge on proving properties of maximizing distributions. Our basic topological space consists of Borel probability measures on $\mathbb{R}^t$ endowed with the weak-convergence topology. This is a metric space with the Levy-Prokhorov metric defining the distance between two probability measures. 

\begin{remark}
For the proofs in this section, it is not necessary to know the precise definition of the metric; but just that the topological space is a metric space and hence normal. Notation wise, most of the time we use random variables $\Xv$ instead of the induced probability measure to represent points on this space. We will also try to state the various theorems that we employ in this section as and when we use them.
\end{remark}

\subsection{Properties of  Additive Gaussian noise}
In this section, we will establish certain properties of distributions obtained according to $Y = X + Z$, where $X$ and $Z$ are independent and $Z \sim \Nc(0,I)$. For simplicity of notation, we consider the scalar case. The authors are confident that these results are known in literature but could not find the relevant sources by a quick Google search.

Let $\tilde{F}(x) = \P(X \leq x)$ (where the inequality is co-ordinate wise. Note that $0 \leq \tilde{F}(x) \leq 1$. Then we see that since $f_z(z)$ has a density, we have
$$ \P(Y \leq y) = \int_{-\infty}^{\infty} \frac{1}{\sqrt{2\pi}} e^{-z^2/2} \tilde{F}(y-z)dz. $$
Thus we have
\begin{align*}
\P(Y \leq y+\delta) = \int_{-\infty}^{\infty} \frac{1}{\sqrt{2\pi}} e^{-z^2/2} \tilde{F}(y+\delta -z)dz 
= \int_{-\infty}^{\infty} \frac{1}{\sqrt{2\pi}} e^{-(z+\delta)^2/2} \tilde{F}(y-z)dz.
\end{align*}
By Dominated convergence theorem stated below (to justify interchange of derivative and integration) $Y$ has a density given by
\begin{align*}
f_Y(y) = \lim_{\delta\to 0} \frac1\delta (\P(Y \leq y+\delta) -\P(Y \leq y)) 
= \int_{-\infty}^{\infty} \frac{-z}{\sqrt{2\pi}} e^{-z^2/2} \tilde{F}(y-z)dz.
\end{align*}
Hence
$$ |f_Y(y)| \leq \int_{-\infty}^{\infty} \frac{|z|}{\sqrt{2\pi}} e^{-z^2/2}dz = \frac{2}{\sqrt{2\pi}}.$$
Again by Dominated convergence theorem we have
$$ f_Y'(y) = \int_{-\infty}^{\infty} \frac{z^2-1}{\sqrt{2\pi}} e^{-z^2/2} \tilde{F}(y-z)dz. $$
Thus
$$ |f_Y'(y)| \leq \int_{-\infty}^{\infty} \frac{|z^2-1|}{\sqrt{2\pi}} e^{-z^2/2}dz \leq 2.$$

\begin{remark}
\label{re:gau}
Thus $Y$ has a bounded density and a bounded first derivative of the density.
In the vector case, similarly we have a bounded density and a uniformly bounded $L_1$ norm for $\nabla f_\Yv(\yv).$
\end{remark}

Next, we state a general lemma which relates weak convergence to convergence of densities.
\begin{lemma}[Lemma 1 in \cite{boo85}]
\label{le:boo}
Suppose that $\Yv_n$ and $\Yv$ have continuous densities $f_n(\yv), f(\yv)$ with respect to the Lebesgue measure on $\mathbb{R}^t$. If $\Yv_n \wc \Yv$ and 
$$ \sup_n |f_n(\yv)| \leq M(\yv) < \infty, ~ \forall \yv \in \mathbb{R}^t$$ and
$$ f_n \text{ is equicontinuous, i.e. } \forall ~\yv, \eps > 0, ~\exists ~\delta(\yv,\eps), n(\yv,\eps)$$
such that $|\yv-\yv_1| < \delta(\yv,\eps)$ implies that $|f_n(\yv) - f_n(\yv_1)| < \eps ~\forall n \geq n(\yv,\eps)$,
then for any compact subset $C$ of $\mathbb{R}^t$
$$ \sup_{\yv \in C} | f_n(\yv) - f(\yv)| \to 0 ~\mbox{as} ~ n \to \infty.$$
If $\{f_n\}$ is uniformly equicontinuous, i.e. $\delta(\yv,\eps)$, $n(\yv,\eps)$ do not depend on $\yv$ and 
$f(\yv_n) \to 0$ whenever $|\yv_n| \to \infty$ then
$$ \sup_{\yv \in \mathbb{R}^t} | f_n(\yv) - f(\yv)| = \|f_n(\yv) - f(\yv)\|_\infty \to 0 ~\mbox{as} ~ n \to \infty.$$
\end{lemma}

\begin{claim}
\label{cl:ecb}
Let $\{\Xv_n\}$ be any sequence of random variables and let $\Yv_n = \Xv_n + \Zv$. Let $f_n(\yv)$ represent the density of $\Yv_n$. Then  the collection of functions
$\{f_n(\yv)\}$ is uniformly bounded and uniformly equicontinuous.
\end{claim}
\begin{proof}
The uniform bound on density is clear from Remark \ref{re:gau}. To see the uniform equicontinuity observe that by mean value theorem
$$| f_n(\yv + \mathbf{\delta}) - f_n(\yv)| =| \nabla f_n(\yv\rq{}) \cdot \mathbf{\delta}| \stackrel{(a)}{\leq} \| \nabla f_n(\yv\rq{})\|_1 \|  \mathbf{\delta}\|_\infty \leq \sqrt{t}  \| \nabla f_n(\yv\rq{})\|_1 \|  \mathbf{\delta}\|_2 $$
where $(a)$ follows from Holder\rq{}s inequality. Now the uniform bound on $L_1$ norm of $\nabla f_\Yv(\yv)$ from Remark \ref{re:gau} yields the desired equicontinuity.
\end{proof}

\begin{definition}
A collection of random variables $\Xv_n$ on $\mathbb{R}^t$ is said to be {\em tight}  if for every $\eps > 0$ there is a compact set $C_\eps \subset \mathbb{R}^t $ such that $\P(\Xv_n \notin C_\eps) \leq \eps,  ~ \forall n$.
\end{definition}
\begin{lemma}
\label{le:smt}
Consider a sequence of random variables $\{\Xv_n\}$ such that $\E(\Xv_n \Xv_n^T) \nsd K, ~ \forall n$. Then the sequence is tight.
\end{lemma}
\begin{proof}
Define $C_\eps = \{ \xv: \|\xv\|_2^2 \leq \frac{tr(K)}{\eps} \}$. By Markov's inequality 
$\P( \|\Xv_n\|^2 >\frac{tr(K)}{\eps} ) \leq \frac{\eps\E(\|\Xv_n\|^2)}{tr(K)} \leq \eps, ~\forall n$.
\end{proof}

\begin{theorem}[Prokhorov]
\label{th:pro}
If $\{\Xv_n\}$ is a tight sequence of random variables in $\mathbb{R}^t$ then there exists a subsequence $\{\Xv_{n_i}\}$ and a limiting probability distribution $\Xv_*$ such that $\Xv_{n_i} \wc \Xv_*$.
\end{theorem}

\begin{lemma}
\label{le:gau}
Let $\Xv_n \wc \Xv_*$ and let $\Zv \sim \Nc(0,I)$ be pairwise independent of $\{\Xv_n\},\Xv_*$. Let $\Yv_n = \Xv_n + \Zv$, $\Yv_* = \Xv_* + \Zv$. Further let $\E(\Xv_n\Xv_n^T) \nsd K, \E(\Xv_* \Xv_*^T) \nsd K$. Let $f_{n}(\yv)$ denote the density of $\Yv_n$ and $f_{*}(\yv)$ denote the density of $\Yv_*$. Then
\begin{enumerate}
\item $\Yv_n \wc \Yv$
\item $f_n(\yv) \to f_*(\yv)$ for all $\yv$
\item $h(\Yv_n) \to h(\Yv).$
\end{enumerate}
\end{lemma}
\begin{proof}
The first part follows from pointwise convergence of characteristic functions (which is equivalent to weak convergence) since
$\Phi_{\Yv_n}(\tv) = \Phi_{\Xv_n}(\tv) e^{-\|\tv\|^2/2}$. The second part (a stronger claim that weak convergence) comes from Lemma \ref{le:boo}. We have uniform equicontinuity  since  $\nabla f_n(\yv)$ has a uniformly bounded  $L_1$ norm (see Remark \ref{re:gau}). Bounded $L_1$ norm of  $\nabla f_n(\yv)$ also implies that $f_*(\yv_n) \to 0$ whenever $|\yv_n| \to \infty$ (A reason: if a point has density $> \eps$ then it has a neighbourhood depending only on $\eps$ where the density is bigger than $\frac{\eps}{2}$, hence this implies that this neighbourhood has a lower bounded probability measure depending only on $\eps$. This cannot happen at infinitely many points of a sequence $\yv_n$ such that $|\yv_n| \to \infty$). The third part comes from Theorem \ref{th:goh}(below) in a direct manner as the densities are uniformly bounded, the second moment($\kappa=2$) is  uniformly bounded by $tr(K)$, and the pointwise convergence from the second part.
\end{proof}

\begin{theorem} [Theorem 1 in \cite{goh04}]
\label{th:goh}
Let $\{\Yv_i \in \mathbb{C}^t\}$ be a sequence of continuous random variables with pdf's $\{f_i\}$ and $\Yv_*$ be a continuous random variable with pdf $f_*$ such that $f_i \to f_*$ pointwise. Let $\|\yv\| = \sqrt{\yv^\dagger\yv}$ denote the Euclidean norm of $\yv \in \mathbb{C}^t$. If $1) \max\{\sup_\yv f_i(\yv), \sup_\yv f_*(\yv) \} \leq F ~ \forall i$ and $2) \max \{ \int \|\yv\|^\kappa f_i(\yv) d \yv, \int \|\yv\|^\kappa f_*(\yv) d \yv\} \leq L$ for some $\kappa > 1$ and for all $i$ then $h(\Yv_i) \to h(\Yv_*)$.
\end{theorem}

\begin{remark}
This theorem is relatively straightforward. One gets $\liminf h(\Yv_i) \geq h(\Yv_*)$ coming due to upper bound on densities and $\limsup h(\Yv_i) \leq h(\Yv_*)$ due to the moment constraints.  Similar kind of result can be found in Appendix 3A of \cite{elk12}.
\end{remark}

We now have the tools to prove Claim \ref{cl:max-S}.
\subsubsection*{Proof of Claim \ref{cl:max-S}}
\begin{proof}
Define
$$\mathsf{v}_\la(K) = \sup_{\Xv:\E(\Xv\Xv^T)=K} \mathsf{s}_\la(\Xv). $$
Let $\Xv_n$ be a sequence of random variables such that $\E(\Xv_n \Xv_n^T) = K$ and $\mathsf{s}_\la(\Xv_n) \uparrow \mathsf{v}_\la(K).$ By the covariance constraint (Lemma \ref{le:smt}) we know that the sequence of random variables $\Xv_n$ forms a tight sequence and by Theorem \ref{th:pro} there exists $X_K^*$ and a convergent subsequence such that $\Xv_{n_i} \stackrel{w}{\Rightarrow} \Xv_K^*$. From Lemma \ref{le:gau} we have that $h(\Yv_{1n_i}), h(\Yv_{2n_i}) \to h(\Yv_{1K}^*), h(\Yv_{2K}^*)$ and hence $\mathsf{s}_\la(\Xv_K^*) = \mathsf{v}_\la(K).$ Thus ${\rm V}_\la(K)$ can be obtained as a convex combination of $\mathsf{s}_{\la}(X^*_K)$ subject to the covariance constraint.

It takes $\frac{t(t+1)}{2}$ constraints to preserve the covariance matrix and one constraint to preserve $s_\la(\Xv|V)$. Hence by Bunt-Carathedory\rq{}s theorem\footnote{We need to use Bunt's extension\cite{bun34} of Caratheodory's theorem as we no longer have compactness of the set required for the usually referred extension due to Fenchel. We can also use vanilla Caratheordory at the expense of one extra cardinality.} we can find a pair of random variables $(V_*,\Xv_*)$ with $|V_*| \leq \frac{t(t+1)}{2} + 1$ such that $ {\rm V}_\la(K) = \mathsf{s}_\la(\Xv_*|V_*). $
\end{proof}

\subsection{Continuity in a pathwise sense on concave envelopes}

In this section we will establish the validity of Claim \ref{cl:max-T}. For this we need more tools and results from analysis.

\begin{claim}
\label{cl:ubs}
For $\la > 1$, there exists $C_\la$ such that $\mathsf{s}_{\la}(\Xv) \leq C_\la$.
\end{claim}
\begin{proof}
We know from Theorem \ref{th:lvalt} that if $\E(\Xv \Xv^T) \nsd K$ then
$$ \mathsf{s}_{\la}(\Xv) \leq S_\la(\Xv) \leq {\rm V}_\la(K) \leq \mathsf{s}_{\la}(\Xv^*_K) $$ for some $\Xv^*_K \sim \Nc(0, K'), K' \nsd K$. This implies that
$$ \sup_{\Xv} \mathsf{s}_{\la}(\Xv) \leq \sup_{K\psd 0:\Xv \sim  \Nc(0, K)} I(\Xv;\Yv_1) - \la I(\Xv;\Yv_2). $$
Let $\Sigma_i = (G_i^TG_i)^{-1}$, $i=1,2$. For $\Xv \sim \Nc(0, K)$, we have
\begin{align*}
2I(\Xv;\Yv_1) - 2\la I(\Xv;\Yv_2) 
&= \log |I+G_1KG_1^T| -\la \log|I+G_2KG_2^T| \\
&= \log |I+KG_1^TG_1| -\la \log |I+KG_2^TG_2| \\
&= -\log|\Sigma_1| +\la\log|\Sigma_2| +\log |\Sigma_1 +K| -\la \log |\Sigma_2 +K| .
\end{align*}
To bound the last two terms, we use the min-max theorem on eigenvalues: Let $\mu_j(A)$ be the $j$-th smallest eigenvalue of symmetric matrix $A\in\mathbb R^{t\times t}$, we have
\begin{align*}
\mu_j(A) = \min_{L_j} \max_{0\neq u\in L_j} \frac{u^TAu}{u^Tu} = \max_{L_{t+1-j}} \min_{0\neq u\in L_{t+1-j}} \frac{u^TAu}{u^Tu},
\end{align*}
where $L_j$ is a $j$ dimensional subspace of $\mathbb R^t$. From this theorem we have
$$\mu_j(K)+\mu_1(\Sigma)\leq \mu_j(K+\Sigma)\leq \mu_j(K) +\mu_t(\Sigma), \quad j=1,2,\ldots,t.$$
Hence
\begin{align*}
\log |\Sigma_1 +K| -\la \log |\Sigma_2 +K| 
&= \sum_{j=1}^t \log \frac{\mu_j(K+\Sigma_1)}{(\mu_j(K+\Sigma_2))^\la} \\
&\leq \sum_{j=1}^t \log \frac{\mu_j(K)+\mu_t(\Sigma_1)}{(\mu_j(K)+\mu_1(\Sigma_2))^{\la}} \\
&\leq t\cdot \log\frac{\mu^* +\mu_t(\Sigma_1)}{(\mu^*+\mu_1(\Sigma_2))^{\la}},
\end{align*}
where $\mu^* = \max\{0,\frac{1}{\la-1} (\mu_1(\Sigma_2)-\la\mu_t(\Sigma_1))\}$.
\end{proof}

For $m\in \mathbb{N}$ the set $\Ac_m:=\{\Xv: \E(\|\Xv\|^2) \leq m \}$ is a closed subset of the topology space. This is because if $\Xv_n \wc \Xv_*$ then $\E(\|\Xv_*\|^2) \leq \liminf_n \E(\|\Xv_n\|^2)$ (by definition of weak convergence and monotone convergence theorem by considering continuous and bounded functions $f_n(x)= \min\{x^2,n\}$.).

We defined $S_\la(\Xv) = \ce(\mathsf{s}_\la(\Xv)) = \sup_{V \to \Xv \to (\Yv_1,\Yv_2)}\mathsf{s}_\la(\Xv|V)$.
Taking $V=\Xv$ we observe that $S_\la(\Xv) \geq 0$. Define $\bar{\mathsf{s}}_\la(\Xv) = \max \{\mathsf{s}_\la(\Xv), 0 \}.$ Now note that $S_\la(\Xv) = \ce(\bar{\mathsf{s}}_\la(\Xv))$, since $S_\la(\Xv) \geq 0$.

Let $ \bar{\mathsf{s}}_\la^m(\Xv)$ be $\bar{\mathsf{s}}_\la(\Xv)$ restricted to $\Ac_m$.  Let $\mathsf{s}_\la^m(\Xv)$ be the continuous extension of $\bar{\mathsf{s}}_\la^m(\Xv) $ from $\Ac_m$ on to $\Pc$. This exists due to Tietze Extension Theorem (produced below).

\begin{theorem}[Tietze Extension Theorem]
\label{th:tietze}
Let $A$ be a closed subset in a normal topological space, then every continuous map $f:A\to \mathbb R$ can be extended to a continuous map on the whole space.
\end{theorem}

  Consider a sequence $\Xv_n \in \Ac_m$ such that $\Xv_n \wc X_*$. Since the second moments are uniformly bounded, similar  arguments as in Claim \ref{cl:max-S} will imply that $\bar{\mathsf{s}}_\la^m(\Xv_n) \to \bar{\mathsf{s}}_\la^m(\Xv_*)$.  Further observe that the function $\mathsf{s}_\la^m(\Xv)$ is bounded and non-negative since $\bar{\mathsf{s}}_\la^m(\Xv)$ is bounded (above by  $C_\la$) and non-negative.

The following result follows from a recent result in \cite{prs09}. The convex hull of a function $f(\Xv)$ is the lower convex envelope, or equivalently $-\ce(-f(\Xv))$, where $\ce(\cdot)$ is the upper concave envelope used in this article.

\begin{theorem}
\label{th:new}
For the set of Borel probability measures on $\mathbb{R}^t$ endowed with the weak-convergence topology, the convex hull of an arbitrary bounded and continuous function is continuous.
\end{theorem}
\begin{proof}
This theorem is obtained directly from Corollary 5 and Theorem 1 in \cite{prs09}.
\end{proof}

An immediate corollary which follows from the fact that convex hull of $f(\Xv) \equiv -\ce(-f(\Xv))$ is the following:
\begin{corollary}
\label{co:new}
For the set of Borel probability measures on $\mathbb{R}^t$ endowed with the weak-convergence topology, the upper concave envelope of an arbitrary bounded and continuous function is continuous.
\end{corollary}

Now define $S_\la^m(\Xv)$ to be concave envelope of $\mathsf{s}_\la^m(\Xv)$. From Corollary \ref{co:new} we have that $S_\la^m(\Xv)$ is continuous; Further since $\mathsf{s}_\la^m(\Xv)$ is bounded, and non-negative, so is $S_\la^m(\Xv)$. Continuity in particular implies that
\begin{equation} \text{if }\Xv_n \wc \Xv_*, ~\mbox{then}~ S_\la^m(\Xv_n) \to S_\la^m(\Xv_*).
\label{eq:eq1}
\end{equation}

\begin{claim} [Continuity in a pathwise sense]
\label{cl:main2}
If $\Xv_n \wc \Xv_*$ and $\E(\Xv_n \Xv_n^T), \E(\Xv_* \Xv_*^T) \nsd K$, then 
$ S_\la(\Xv_n) \to S_\la(\Xv_*)$.
\end{claim}
\begin{proof}
The proof is essentially validating the  interchange of limits between $m,n$ in \eqref{eq:eq1}.
We show a uniform convergence (in $m$) of $S_\la^m(\Xv_n) \to S_\la(\Xv_n)$ and this suffices to justify the interchange as follows:
Given $\eps > 0$ choose $M_\eps>0$ such that $|S_\la(\Xv_n) - S_\la^m(\Xv_n)| < \eps ~ \forall n$ whenever $m > M_\eps$ (such an $M_\eps$ exists by uniform convergence). This implies that $\forall  m >  M_\eps$ we have
$$ S_\la(\Xv_n) \leq S_\la^m(\Xv_n) + \eps, \stackrel{n \to \infty}{\implies} \limsup_n S_\la(\Xv_n) \leq S_\la^m(\Xv_*) + \eps, \stackrel{m \to \infty}{\implies}
\limsup_n S_\la(\Xv_n) \leq S_\la(\Xv_*) + \eps.$$
Similarly $\forall  m >  M_\eps$
$$ S_\la(\Xv_n) \geq S_\la^m(\Xv_n) - \eps, \stackrel{n \to \infty}{\implies} \liminf_n S_\la(\Xv_n) \geq S_\la^m(\Xv_*) - \eps, \stackrel{m \to \infty}{\implies}
\liminf_n S_\la(\Xv_n) \geq S_\la(\Xv_*) - \eps.$$

Hence $ S_\la(\Xv_n) \to S_\la(\Xv_*)$ provided we show the uniform convergence (in $m$) of $S_\la^m(\Xv_n) \to S_\la(\Xv_n)$. Given $\eps > 0$ consider a $V$ such that $S_\la(\Xv_n) \leq s_\la(\Xv_n|V) + \frac{\eps}{4}$. Observe that $V$ induces a probability measure on the space of all probability measures. We now bound the induced probability measure on distributions
such that $\E(\|\Xv\|^2) \geq m$. Since $\E(\|\Xv_n\|^2) \leq tr(K)$, from Markov's inequality the mass of the induced measure on the probability measures such that $\E(\|\Xv\|^2) \geq m$ is at most $\frac{tr(K)}{m}$. Hence their contribution to $s_\la(\Xv_n|V)$ is at most $\frac{C_\la tr(K)}{m}$, where $C_\la$ is the global upper bound on $\mathsf{s}_\la(\Xv)$. Thus by taking $m$ large enough we can make this smaller than $\frac{\eps}{4}$. Hence 
$$S_\la^m(\Xv_n) \geq s_\la^m(\Xv_n|V) \geq s_\la(\Xv_n|V) -  \frac{\eps}{4} \geq  S_\la(\Xv_n) -\frac{\eps}{2}.$$
Similar argument (taking $V\rq{}$ such that $S_\la^m(\Xv_n) \leq s_\la^m(\Xv_n|V\rq{}) + \frac{\eps}{4}$)  also shows that $ S_\la(\Xv_n) \geq S_\la^m(\Xv_n) - \frac{\eps}{2}.$ Hence for all $m > \frac{4 C_\la tr(K)}{\eps}$ we have that $|S_\la(\Xv_n) - S_\la^m(\Xv_n)| \leq \eps $ uniformly in $n$ as desired.
\end{proof}

We now have the tools to prove Claim \ref{cl:max-T}.
\subsubsection*{Proof of Claim \ref{cl:max-T}}

\begin{proof}
From Claim \ref{cl:main2} and using similar arguments as in the proof of Claim \ref{cl:max-S} we see that ${\rm\hat V}_\lav(K)$ can be obtained as a convex combination of $\mathsf{t}_{\lav}(X^*_K)$ subject to the covariance constraint. It takes $\frac{t(t+1)}{2}$ constraints to preserve the covariance matrix and one constraint to preserve $\mathsf{t}_\lav(\Xv|W)$. Hence by Bunt-Carathedory\rq{}s theorem we can find a pair of random variables $(W_*,\Xv_*)$ with $|W_*| \leq \frac{t(t+1)}{2} + 1$ such that $ {\rm\hat V}_\lav(K) = \mathsf{t}_\lav(\Xv_*|W_*). $
\end{proof}

Indeed the proof technique we used carries over almost verbatim to establish this general lemma, which could be useful in other multi-terminal situations..
\begin{lemma}
Consider the space of all Borel probability distributions on $\mathbb{R}^t$ endowed with the topology induced by weak convergence.
If $f(\Xv)$ is a bounded real-valued function with the following property, \textrm{P}: for any sequence $\{\Xv_n\}$ that satisfies the two properties $(i)$ $\exists~ \kappa > 1, ~  s/t ~ \E(|\Xv_n|^\kappa) \leq B ~\forall n$ (i.e. sequence has a uniformly bounded $\kappa$-th moment) and $(ii)$ $\Xv_n \wc \Xv_*$, we have
$f(\Xv_n) \to f(\Xv_*)$; then the same properties holds for $F(\Xv)=\ce(f(\Xv))$, its upper concave envelope; i.e. $F(\Xv)$  is bounded and satisfies \textrm{P}.
\end{lemma}
\begin{proof}
The boundedness of $F(\Xv)$ is immediate. To show that $F(\Xv)$ satisfies property P, we use the same argument as earlier. Consider a sequence  $\{\Xv_n\}$ with a uniformly bounded $\kappa$-th moment  such that $\Xv_n \wc \Xv_*$. First, restrict $f$ to $\Ac_m$ (set of all distributions whose $\kappa$-th moment is upper bounded by $m$) and observe that this induces is a continuous (by property $P$ of $f$) and bounded function (on the topology induced by weak convergence)  from this closed set, $\Ac_m$, to reals. Now we extend this restricted function by the Tietze extension theorem to obtain $f^m(\Xv)$, a continuous and bounded function on the whole space. Then from Corollary \ref{co:new} we see that the concave envelope of $f^m(\Xv)$, denoted by $F^m(\Xv)$ is bounded and continuous. Finally one can establish a uniform convergence (in $n$) of $F^m(\Xv_n) \to F(\Xv_n)$ and hence conclude that $F(\Xv_n) \to F(\Xv_*) $.
\end{proof}

\section{Alternate path to Theorem \ref{th:lvalt}}
\label{se:altclt}
Below, we will give an elementary proof of Theorem \ref{th:lvalt} without invoking Corollary \ref{co:tworv}.

\begin{corollary}
\label{co:ind}
For every $l \in \mathbb{N}, n=2^{l},$ let $(V^n,\Xv_n) \sim \prod_{i=1}^n p_*(V_i, \Xv_i).$ Then $\Vt, \tilde{\Xv}_n$ achieves ${\rm V}_\la(K)$ where $\Vt = (V_1, V_2,.., V_n)$ and $\tilde{\Xv}_n|\big(\Vt = (v_1, v_2,..,v_n)\big)  \sim  \frac{1}{\sqrt{n}} \left( \Xv_{v_1} + \Xv_{v_2} + \cdots + \Xv_{v_n} \right).$ We take $\Xv_{v_1}, \Xv_{v_2}, \ldots, \Xv_{v_n}$ to be independent random variables here.
\end{corollary}
\begin{proof}
The proof follows from induction using Claim \ref{cl:twolet-S}.
\end{proof}

Consider $(V^n,\Xv^n) \sim \prod_{i=1}^n p_*(V_i, \Xv_i),$ where $p_*(v,\xv)$ achieves $\textrm{V}_\la(K)$. Let $\Vc=\{1,..,m\}$ where $m \leq \frac{t(t+1)}{2} + 1$. Now consider $(V^n, \tilde{\Xv}_n)$ where $\tilde{\Xv}_n|\big(V^n = (v_1, v_2,..,v_n)\big)  \sim  \frac{1}{\sqrt{n}} \left( \Xv_{v_1} + \Xv_{v_2} + \cdots + \Xv_{v_n} \right).$ Again we take $\Xv_{v_1}, \Xv_{v_2}, \ldots, \Xv_{v_n}$ to be independent random variables.

 As is common in information theoretic arguments, we are going to consider typical sequences and atypical sequences. Let us define typical sequences in the following fashion:
$$ \caep(V) := \{v^n: \big| |\{i:v_i = v\}| - np_*(v) \big| \leq n \omega_n p_*(v), ~ \forall v \in [1:m]. \} $$
where $\omega_n$ is any sequence such that $\omega_n \to 0$ as $n \to \infty$ and $\omega_n \sqrt{n} \to \infty$ as $n \to \infty$. For instance $\omega_n = \frac{\log n}{\sqrt n}$.

Note that (using Chebychev\rq{}s inequality)
 $$\P( \big| |\{i:v_i = v\}| - np_*(v) \big| > n \omega_n p_*(v)) \leq \frac{1-p_*(v)}{p_*(v)\omega_n^2 n}.$$
Hence $\P(v^n \notin \caep(V)) \to 0$ as $n \to \infty$.

Consider any sequence of typical sequences  $v^n \in \caep(V)$. Consider a sequence of induced distributions $\hat{\Xv}_n \sim \tilde{\Xv}_n|v^n$.
\begin{claim}
\label{cl:clcon}
$ \hat{\Xv}_{n} \Rightarrow \Nc(0, \sum_{v=1}^m p_*(v) K_v )$

\end{claim}
\begin{proof}
For  given $v^n$, let $A_n(v) =  |\{i:v_i = v\}|$. We know that $A_n(v) \in np_*(v) (1 \pm w_n), \forall v$.  Consider a $\cv \in \mathbb{R}^{t}$ with $\|\cv\| = 1$. Let
$ \hat{\Xv}_{n,i}^\cv \sim \frac{1}{\sqrt{n}}\cv^T \cdot \Xv_{v_i}$ and $\hat{\Xv}_{n,i}^\cv$ be independent random variables. Note that $\sum_{i=1}^n  \hat{\Xv}_{n,i}^\cv  \sim \cv^T  \hat{\Xv}_{n}$.
\smallskip

Note that 
\begin{align*}
  \sum_{i=1}^n E( (\hat{\Xv}_{n,i}^\cv)^2) &= \frac{1}{n} \sum_{v} A_n(v) \cv^T K_v \cv \to \cv^T \left(\sum_v p_*(v) K_v\right) \cv.  \\
 \sum_{i=1}^n E( (\hat{\Xv}_{n,i}^\cv)^2; |\hat{\Xv}_{n,i}^c| > \eps_1) &= \frac{1}{n} \sum_{v} A_n(v)\E(\cv^T \Xv_v \Xv_v^T \cv ;  \cv^T \Xv_v \Xv_v^T \cv  \geq n \eps_1^2) \\
& \leq \sum_v p_*(v)(1 + \omega_n) \E(\cv^T X_v X_v^T \cv ;  \cv^T \Xv_v \Xv_v^T \cv  \geq n \eps_1^2) \to 0.
\end{align*}
In the last convergence we use that $K_v$\rq{}s are bounded, and hence $\cv^T \Xv_v$ has a bounded seconded moment.
Hence from Lindeberg-Feller CLT\footnote{We adopt the notation in Theorem (4.5), Chapter 2 in \cite{dur96}.} we have $\sum_{i=1}^n \hat{\Xv}_{n,i}^\cv \Rightarrow \Nc(0,\cv^T \sum_v p_*(v) K_v \cv)$. Hence $\hat{\Xv}_n \Rightarrow \Nc(0, \sum_v p_*(v) K_v )$ (Cramer-Wold device).
\end{proof}

The next claim shows a  uniform convergence of the conditional laws to the Gaussian.

\begin{claim}
\label{cl:uc}
Given any $\delta > 0$, there exists $N_0$ such that $\forall n > N_0$ we have for all $v^n \in \caep(V)$
$$ \mathsf{s}_\la(\tilde{\Xv}_n|v^n) - \mathsf{s}_\la(\Xv^*) \leq \delta, $$
where $\Xv^* \sim \Nc(0, \sum_v p_*(v) K_v).$
\end{claim}
\begin{proof}
Assume not. Then we have a subsequence $v^{n_k}\in \mathcal{T}^{(n_k)}(V)$ and distributions $\tilde{\Xv}_{n_k}|v^{n_k}$ such that
$$ \mathsf{s}_\la(\tilde{\Xv}_{n_k}|v^{n_k}) > \mathsf{s}_\la(X^*) + \delta, \forall k. $$
However from Claim \ref{cl:clcon} we know that $\tilde{\Xv}_{n_k}|v^{n_k} \stackrel{w}{\Rightarrow} X^*$ and from Lemma \ref{le:gau} we have  $\mathsf{s}_\la(\tilde{\Xv}_{n_k}|v^{n_k}) \to \mathsf{s}_\la(X^*)$, a contradiction.
\end{proof}

\medskip

\begin{theorem}
\label{th:main}
There is a single Gaussian distribution (i.e. no mixture is required) that achieves ${\rm V}_\la(K)$.
\end{theorem}
\begin{proof}
We know from Corollary \ref{co:ind} that For every $l \in \mathbb{N}, n=2^{l},$ the pair $V^n, \tilde{\Xv}_n$ achieves $\textrm{V}_\la(K)$.
Hence 
$$ {\rm V}_\la(K) = \sum_{v^n} p_*(v^n) \mathsf{s}_\la(\tilde{\Xv}_n|v^n) = \sum_{v^n \in \caep(V)} p_*(v^n)\mathsf{s}_\la(\tilde{\Xv}_n|v^n)  + \sum_{v^n \notin \caep(V)} p_*(v^n)\mathsf{s}_\la(\tilde{\Xv}_n|v^n).$$
For a given $v^n$, let $\hat{\Xv} \sim \Xv_n|v^n$. Then note that $\E(\hat{\Xv}\hat{\Xv}^T) \nsd \sum_{v=1}^m K_v.$ Thus $\mathsf{s}_\la(\hat{\Xv}) \leq I(\hat{\Xv};\Yv_1) \leq C$ for some fixed constant $C$ that is independent of $v^n$. Thus using Claim \ref{cl:uc} we can upper bound ${\rm V}_\la(K)$ for large $n$ by
\begin{align*}
 {\rm V}_\la(K) & =  \sum_{v^n \in \caep(V)} p_*(v^n)\mathsf{s}_\la(\tilde{\Xv}_n|v^n)  + \sum_{v^n \notin \caep(V)} p_*(v^n)\mathsf{s}_\la(\tilde{\Xv}_n|v^n)\\
 & \leq \sum_{v^n \in \caep(V)} p_*(v^n)( \mathsf{s}_\la(\Xv^*) + \delta) + C  \sum_{v^n \notin \caep(V)} p_*(v^n) \\
 & = \P(v^n \in \caep) ( \mathsf{s}_\la(\Xv^*) + \delta) + C \P(v^n \notin \caep).
\end{align*}
Here $\Xv^* \sim \Nc(0, \sum_v p_*(v) K_v).$ Since $ \P(v^n \in \caep) \to 1$ as $ n \to \infty$ we get ${\rm V}_\la(K) \leq  \mathsf{s}_\la(\Xv^*) + \delta$; but $\delta > 0$ is arbitrary, hence
${\rm V}_\la(K) \leq \mathsf{s}_\la(\Xv^*)$. The other direction ${\rm V}_\la(K) \leq \mathsf{s}_\la(\Xv^*)$ is trivial from the definition of ${\rm V}_\la(K)$ and the fact that $\sum_{v} p_*(v) K_v \nsd K$.
\end{proof}

\end{document}

%% file: Chead.tex
\usepackage{geometry} 
\usepackage{graphicx} 
\usepackage{color}

\definecolor{Red}{rgb}{1,0,0}
\definecolor{Blue}{rgb}{0,0,1}
\definecolor{Olive}{rgb}{0.41,0.55,0.13}
\definecolor{Green}{rgb}{0,1,0}
\definecolor{MGreen}{rgb}{0,0.8,0}
\definecolor{DGreen}{rgb}{0,0.55,0}
\definecolor{Yellow}{rgb}{1,1,0}
\definecolor{Cyan}{rgb}{0,1,1}
\definecolor{Magenta}{rgb}{1,0,1}
\definecolor{Orange}{rgb}{1,.5,0}
\definecolor{Violet}{rgb}{.5,0,.5}
\definecolor{Purple}{rgb}{.75,0,.25}
\definecolor{Brown}{rgb}{.75,.5,.25}
\definecolor{Grey}{rgb}{.5,.5,.5}

\usepackage{booktabs} 
\usepackage{array} 
\usepackage{paralist} 
\usepackage{verbatim} 
\usepackage{subfigure} 
\usepackage{amsmath,amssymb,amsthm}

\usepackage{fancyhdr} 
\theoremstyle{plain}

\newtheorem{theorem}{Theorem} 
 
\newtheorem{corollary}{Corollary}
\newtheorem{claim}{Claim}
\newtheorem{lemma}{Lemma}

\newtheorem{bound}{Bound}

\theoremstyle{remark}

\newtheorem{remark}{Remark}

\theoremstyle{definition}
\newtheorem{definition}{Definition}

%% file: defns.tex
\usepackage{xspace}
\usepackage{mathrsfs}
\usepackage{amsopn}

\newcommand{\Ac}{\mathcal{A}}

\newcommand{\Cc}{\mathcal{C}}

\newcommand{\Ic}{\mathcal{I}}

\newcommand{\Nc}{\mathcal{N}}
\newcommand{\Oc}{\mathcal{O}}
\newcommand{\Pc}{\mathcal{P}}

\newcommand{\Vc}{\mathcal{V}}

\newcommand{\Xv}{{\bf X}}
\newcommand{\Yv}{{\bf Y}}
\newcommand{\Zv}{{\bf Z}}

\newcommand{\xv}{{\bf x}}

\newcommand{\tv}{{\bf t}}
\newcommand{\yv}{{\bf y}}

\newcommand{\Vt}{{\tilde{V}}}
\newcommand{\Wt}{{\tilde{W}}}

\newcommand{\vt}{{\tilde{v}}}
\newcommand{\wt}{{\tilde{w}}}

\def\eps{\epsilon}
\def\la{\lambda}

\def\ce{\mathfrak{C}}

\DeclareMathOperator\E{E}
\let\P\relax
\DeclareMathOperator\P{P}

\def\textiid{i.i.d.\@\xspace}
\newcommand\iid{\ifmmode\text{ i.i.d. } \else \textiid \fi}

\newcommand{\Real}{\mathbb{R}}

\newcommand{\qmf}{\mathfrak{q}}
